\documentclass[journal,twoside,web]{ieeecolor}
\usepackage{generic}
\usepackage{cite}
\usepackage{amsmath,amssymb,amsfonts}
\usepackage{algorithmic}
\usepackage{graphicx}
\usepackage[latin9]{inputenc}
\usepackage[OT2,T1]{fontenc}
\usepackage{MnSymbol}
\usepackage{lipsum}

\newtheorem{theorem}{Theorem}
\newtheorem{assumption}{Assumption}
\newtheorem{definition}{Definition}
\newtheorem{property}{Property}
\newtheorem{proposition}{Proposition}

\newtheorem{lemma}{Lemma}
\newtheorem{corollary}[lemma]{Corollary}
\newtheorem{remark}{Remark}
\DeclareSymbolFont{cyrletters}{OT2}{wncyr}{m}{n}
\DeclareMathSymbol{\Sha}{\mathalpha}{cyrletters}{"58}
\usepackage{textcomp}
\def\BibTeX{{\rm B\kern-.05em{\sc i\kern-.025em b}\kern-.08em
 T\kern-.1667em\lower.7ex\hbox{E}\kern-.125emX}}
\begin{document}

\title{Necessary and Sufficient Conditions for Harmonic Control in Continuous Time}
\author{N. Blin$^{1,2,3}$, P. Riedinger$^{1}$, J. Daafouz$^{1}$, L. Grimaud$^2$, P. Feyel$^{3}$
\thanks{$^{1}$Universit\'e de Lorraine, CNRS-CRAN UMR 7039, 2, avenue de la for\^et de Haye, 54516 Vandoeuvre-l\`es-Nancy Cedex, France.}
\thanks{$^2$ Safran Electronics \& Defense, Research \& Technologies department, 95 route de Montpellier, Valence, France}
\thanks{$^3$ Safran Electronics \& Defense, Research \& Technologies department, 100 avenue de Paris, Massy, France}
\thanks{This work is supported by HANDY project ANR-18-CE40-0010-02.} 
}
\maketitle

\begin{abstract} 
In this paper, we revisit the concepts and tools of harmonic analysis and control and provide a rigorous mathematical answer to the following question: when does an harmonic control has a representative in the time domain ? By representative we mean a control in the time domain that leads by sliding Fourier decomposition to exactly the same harmonic control. Harmonic controls that do not have such representatives lead to erroneous results in practice. The main results of this paper are: a one-to-one correspondence between ad hoc functional spaces guaranteeing the existence of a representative, a strict equivalence between the Carathéorody solutions of a differential system and the solutions of the associated harmonic differential model, and as a consequence, a general harmonic framework for Linear Time Periodic (LTP) systems and bilinear affine systems. The proposed framework allows to design globally stabilizing harmonic control laws. We illustrate the proposed approach on a single-phase rectifier bridge. Through this example, we show how one can design stabilizing control laws that guarantee periodic disturbance rejection and low harmonic content.
\end{abstract}
\begin{IEEEkeywords}
	Sliding Fourier Decomposition, Dynamic phasors, Harmonic modeling and control, repetitive control, Lyapunov harmonic equations, Riccati harmonic equations, bilinear affine systems, power converters
\end{IEEEkeywords}

\section{Introduction}
Designing controllers that achieve both good dynamic performance and a low harmonic distortion at steady state is a very challenging problem from both a theoretical and practical point of view. This problem arises in many applications and in particular those related to power systems such as transmission and conversion of electrical power where more and more complex power converters with non-linear components and complicated frequential behaviors are involved. A central concern is that of modeling harmonic behaviors and developing harmonic based control methods with stability and performance guarantees often expressed as a reduction of harmonic distortion. To this end, many harmonic modeling approaches have been introduced in the literature, such as generalized state-space averaging \cite{Middlebrook}, \cite{Sanders}, dynamic phasors \cite{Mattavelli}, \cite{Almer}, extended harmonic domain \cite{Madrigal}, harmonic state-space \cite{Wereley_1990}, \cite{Ormrod}, \cite{Hwang} etc. The models introduced in these papers capture both the transient evolution and the steady-state of harmonics and are widely used for power systems analysis and control, e.g. \cite{Mollerstedt}, \cite{Javaid}, \cite{Mattavelli}, \cite{Stankovic}, \cite{Tadmor}, \cite{Almer2}, \cite{Chavez}, \cite{Karami}, \cite{Rico}, \cite{Ghita}. We also mention \cite{Zhou2002}, \cite{Zhou2008}, \cite{Bittanti} where control design methods based on harmonic Lyapunov or Riccati equations are proposed. Repetitive control methods are also an
alternative that allow to track periodic signals and/or reject
periodic disturbances (see \cite{Hara}, \cite{Hillerstrom}, \cite{Steinbuch} \cite{Ghosh} \cite{Longman} and \cite{Wang}.) but our focus in this paper is on harmonic methods.

The main feature of the harmonic modeling techniques is that a linear time-periodic model can be regarded, in the harmonic domain, as a linear time-invariant model, possibly with a state space of infinite dimension. This feature allows to ease the analysis and control design tasks as classical Linear Time Invariant (LTI) methods can be applied. We refer to \cite{EJC2020} where a detailed review is given and where the links between these different approaches are highlighted. However, a fundamental question remains unanswered : when does an harmonic control designed using these approaches has a representative control in the time domain ? By representative we mean a control in the time domain that leads by sliding Fourier decomposition to exactly the same harmonic control. Harmonic controls that do not have such representatives are of no interest in practice. Indeed, even if these harmonic controls stabilize the harmonic model, one cannot obtain their time domain counterparts. This is a very important question as it has a major impact on the implementation of the designed control laws and their stabilizing properties. To the best of our knowledge, there is no result in the literature that answers this question in a rigorous mathematical framework. The objective of this paper is to fill this gap.

The main tool used in harmonic modeling is the the sliding Fourier decomposition. Following the example of Fourier series for which Riesz-Fisher theorem establishes a one-to-one correspondence between the square integrable functions and square summable sequences, we point out the fact that it is necessary to specify precisely the functional spaces involved in the harmonic methods. This is a key point that allows us to build the main contributions of this paper and answer the previous question.

This paper is organized as follows. First, we start by recalling in the next section some of the concepts that appear at the intersection of signal processing and control theory and which will be used in the paper. This concerns the Fourier series, the sliding Fourier decomposition as well as the ad hoc functional spaces. In section III, we present the first main result of this paper which is a necessary and sufficient condition called "coincidence condition" that guarantees a one-to-one correspondence between temporal signals and a subspace of harmonic signals denoted $H$. As it well be explained in the sequel, this result has important consequences especially in control design. In particular, it implies that any transformation performed on any variable in the harmonic domain must have an equivalent transformation in the time domain, otherwise the correspondence is lost. Moreover, this section provides the ingredients to guarantee that the inverse of a sliding Fourier decomposition recovers the original signal in $L^2$ sense. A necessary and sufficient condition that characterizes harmonic signals belonging to $H$ is also provided. The second main result of this paper is presented in Section IV. It concerns a necessary and sufficient condition to have a strict equivalence between solutions (in Carathéodory sense) of general differential equations and solutions of the associated harmonic ones. This result is generic and it has a major impact. We illustrate this impact in section V on stability and stabilization of very general Linear Time Periodic (LTP) systems and bilinear affine systems. We first prove a strict equivalence between harmonic Lyapunov or Riccati equations and their periodic versions under very weak assumptions on the regularity of the time varying matrices. This extends existing results to more general LTP systems. Indeed, this equivalence is established independently of the sliding window period $T$ without invoking Floquet theory \cite{Floquet} as done in general \cite{Zhou2008}, \cite{Bittanti}. As a consequence, a Lyapunov based approach is proposed for global asymptotic stability analysis and control design for both LTP systems and bilinear affine systems subject to periodic exogenous inputs and disturbances. The proposed harmonic framework allows the design of globally and asymptotically stabilizing periodic control laws suitable for both tracking and periodic disturbance rejection. Finally, in Section VI, we apply the control design approach proposed in this paper to a single-phase rectifier bridge, an electronic device that converts AC supply voltage to DC one. Through this example, we illustrate how one can design stabilizing control laws that guarantee periodic disturbance rejection and low harmonic content. We end the paper by a conclusion and perspectives.

{\bf Notation:} The transpose of a matrix $A$ is denoted $A'$ and $A^*$ denote the complex conjugate transpose $A^*=\bar A'$. The $n$-dimensional identity matrix is denoted $Id_n$. $A \otimes B$ is the Kronecker product of two matrices $A$ and $B$. The space of piecewise $k$-differentiable and continuous functions is denoted $C^k_{pw}$. $C^a$ is the space of absolutely continuous function and $L^{p}$ the Lebesgues spaces for $1\leq p\leq\infty$. $x(t^-)$ and $x(t^+)$ denote respectively the left and right limits of $x$ at time $t$. The notation $a.e.$ means almost everywhere. 

\section{Preliminaries on functional space and sliding Fourier Decomposition}

For a given interval $[a,b]$, we recall that the functional space $L^2([a,b],\mathbb{C}^n)$ is the space of measurable and square integrable functions on $[a,b]$ with values in $\mathbb{C}^n$
endowed with the scalar product $$<x,y>_{L^2([a,b],\mathbb{C}^n)}=\frac{1}{b-a}\int_{a}^bx(\tau)^*y(\tau)d\tau,$$
and of induced norm:
$\|x\|_{L^2([a,b],\mathbb{C}^n)} =<x,x>_{L^2([a,b],\mathbb{C}^n)}^{\frac{1}{2}}.$

We also consider the space of complex square summable sequences: $$\ell^2(\mathbb{C}^n)=\{ X : k\in\mathbb{Z}\mapsto X_k\in \mathbb{C}^n\text{ such that} \sum_{k=-\infty}^{+\infty} X_k^*X_k<+\infty\},$$ endowed with the scalar product $$<X,Y>_{\ell^2(\mathbb{C}^n)}=\sum_{k=-\infty}^{+\infty}X_k^*Y_k,$$ and of induced norm:
$\|X\|_{\ell^2(\mathbb{C}^n)}=<X,X>_{\ell^2(\mathbb{C}^n)}^{\frac{1}{2}}$.\\
We recall here two Theorems of particular importance. The first one states that there is an isometry between the spaces $L^2([a,b],\mathbb{C}^n)$ and $\ell^2(\mathbb{C}^n)$. The second one states that the Fourier decomposition defines a one-to-one correspondence (bijective function) between these two spaces.
\begin{theorem}[Parseval's Identity]\label{thparseval} $\forall x\in L^2([a,b],\mathbb{C}^n)$, the following relation holds:$$\|x\|_{L^2([a,b],\mathbb{C}^n)}= \|X\|_{\ell^2(\mathbb{C}^n)},$$ with the components $X_k$ given by 
$$X_k=\frac{1}{b-a}\int_{a}^b x(\tau)e^{-j\omega k \tau}d\tau,$$
and where $\omega =\frac{2\pi}{b-a}$.
\end{theorem}

\begin{theorem}[Riesz-Fischer's Theorem]\label{riesz}
A function $x$ is square integrable if and only if the corresponding Fourier series defined by: 
$$y(t)=\sum_{k=-\infty}^{+\infty} X_k e^{j \omega k t}$$
converges to $x$ in the space $L^2$.
\end{theorem}

In the sequel, we will often consider the following functional spaces.
 \begin{definition}For a given $p$, $1\leq p\leq +\infty$,
 a function $x$ is a locally $p$ integrable function denoted by $x\in L^p_{loc}(\mathbb{R},\mathbb{C}^n)$ if for any compact set $A$ of $\mathbb{R}$, its restriction $x|_A\in L^p(A,\mathbb{C}^n)$.
\end{definition}

%

We also recall the definition of the sliding Fourier decomposition over a window of length $T$:
\begin{definition}The sliding Fourier decomposition over a window of length $T$ from $ L^{2}_{loc}(\mathbb{R},\mathbb{C}^n)$ to $L^{\infty}_{loc}(\mathbb{R},\ell^2(\mathbb{C}^n))$ is defined by:
\begin{align*}
\mathcal{F}_T: \ &L^{2}_{loc}(\mathbb{R},\mathbb{C}^n) \rightarrow L^{\infty}_{loc}(\mathbb{R},\ell^2(\mathbb{C}^n))\\
&x\mapsto X 
\end{align*}
where the time varying infinite sequence $X$ is defined by:
$$t\mapsto \mathcal{F}_T(x)(t)=X(t)$$
whose ($n$-dimensional) components $X_k(t)$, $k\in\mathbb{Z}$ satisfy:
\begin{align*}X_k(t)& =\frac{1}{T}\int_{t-T}^t x(\tau)e^{-j\omega k \tau}d\tau.
\end{align*}
 \end{definition}

In this definition, the time varying component $X_k$ is often called $k-$th phasor of the signal $x$. These phasors can be viewed as the result of a sliding projection of $x$ on the Fourier basis. 
It can be noticed that if $x\in L_{loc}^{2}(\mathbb{R},\mathbb{C}^n)$, the restriction of $x$ to the interval $[t-T,t]$ denoted by $x|_{ [t-T,t]}$ belongs obviously to $L^2([t-T,t],\mathbb{C}^n)$.
Moreover, $X$ is well defined in $L^{\infty}_{loc}(\mathbb{R},\ell^2(\mathbb{C}^n))$ as the functions $X_k$ are absolutely continuous with respect to $t$ (thus bounded on every compact set) and since, for any fixed $t$, the sequence of terms $X_k(t)$ is square summable (see Parseval's Identity~\ref{thparseval}).

Note also that the definition of the sliding Fourier decomposition can be extended to more general signals belonging to the space $L^{1}_{loc}(\mathbb{R}, \mathbb{C}^n)$ but for which the space of arrival is not any more $L^{\infty}_{loc}(\mathbb{R},\ell^2(\mathbb{C}^n))$ but $L^{\infty}_{loc}(\mathbb{R},\ell^{\infty}(\mathbb{C}^n))$.

In the sequel and for a given $p=1,2$ or $\infty$, by abuse of language, we will say that a function defined on $\mathbb{R}$ belongs to $L^p(I,\mathbb{C}^n)$ if its restriction to the interval $I$, $x|_{ I}\in L^p(I,\mathbb{C}^n)$.
Moreover, in order to simplify the notations, $L^p([a,b])$ or $L^p$ will be often used instead of $L^p([a,b],\mathbb{C}^n)$. 
Thus, for example, $x\in L^2([a,b])$ means $x \in L^2([a,b],\mathbb{C}^n)$.
\section{Coincidence Condition and Inverse Sliding Fourier decomposition}
Before detailing the contributions related to harmonic control, we present in this section a preliminary theoretical result of particular importance. It concerns conditions under which a one-to-one correspondence can be established between time domain signals and harmonic domain ones. This result has a strong impact on the contributions related to solutions of general dynamical systems, their stability analysis and the design of stabilizing harmonic control laws.

Let $X\in L_{loc}^{\infty}(\mathbb{R},\ell^2)$. By application of Riesz-Fisher's Theorem (Thereom~\ref{riesz}), it is possible for almost every $t$ to consider a function $x_t\in L^2([t-T,t])$ such that for $\tau\in [t-T,t]$,
$$x_t(\tau)=\sum_{k=-\infty}^{+\infty} X_k(t)e^{j\omega k \tau},$$ where
$X_k(t)=\frac{1}{T}\int_{t-T}^t x_t(\tau)e^{-j\omega k \tau}d\tau.$

However, if we consider two time instants $t_1<t_2$ with $t_2-t_1<T$, the functions $x_{t_1}$ and $x_{t_2}$ do not necessarily coincide on their common definition domain, that is we do not have:
$$x_{t_1}(\tau)= x_{t_2}(\tau),$$
for almost all $ \tau\in[t_2-T, \ t_1]$ (see Fig. \ref{f1} and \ref{f2}).
Thus, there does not exist a priori a function $x \in L^{2}_{loc}(\mathbb{R},\mathbb{C}^n)$ such that 
$$X=\mathcal{F}_T(x).$$

 \begin{figure}[h]
\begin{minipage}{0.48\linewidth}

\begin{center}
\includegraphics[width=\linewidth,height=3cm]{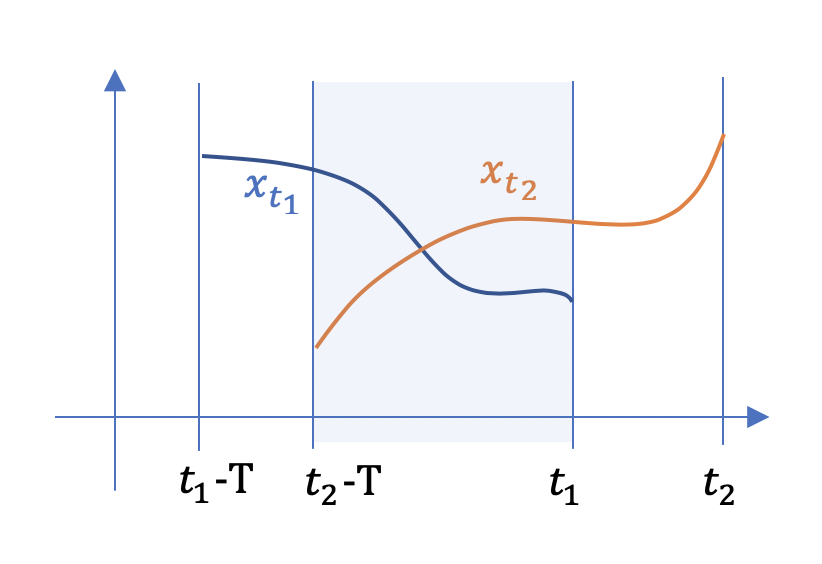}
\caption{The functions $x_{t_1}$ and $x_{t_2}$
do not coincide on the interval $[t_2-T, \ t_1]$}\label{f1}
\end{center}
\end{minipage}\hfill
\begin{minipage}{0.48\linewidth}
\begin{center}
\includegraphics[width=\linewidth,height=3cm]{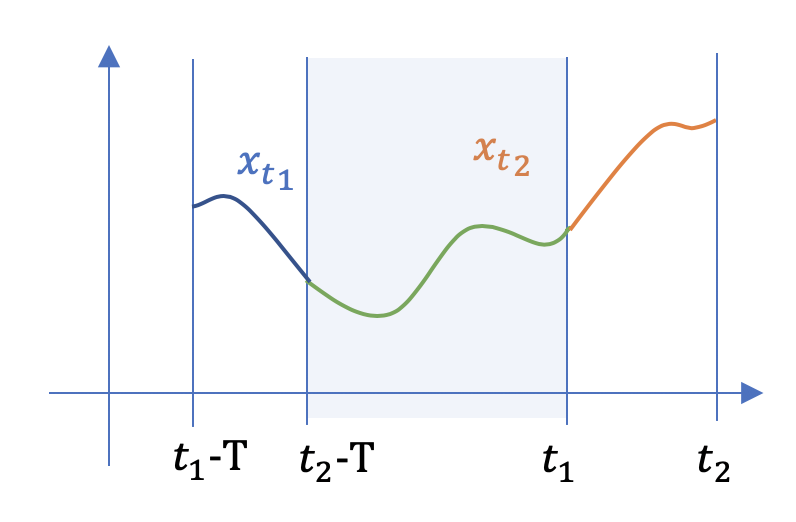}
\caption{Coincidence situation: $x_{t_1}(\tau)= x_{t_2}(\tau)$
for almost all $ \tau\in[t_2-T, \ t_1]$}\label{f2}
\end{center}
\end{minipage}
\end{figure}

We are now in position to define the notion of a representative of $X$ and state the main result of this section.
\begin{definition}Let $X\in L_{loc}^{\infty}(\mathbb{R},\ell^2(\mathbb{C}^n))$. A function $x\in L_{loc}^{2}(\mathbb{R},\mathbb{C}^n)$ is a representative of $X$ if $X=\mathcal{F}_T(x)$.
\end{definition}

\begin{theorem}[Coincidence Condition]\label{coincidence}
There exists a representative $x$ of $X$, with $X\in L_{loc}^{\infty}(\mathbb{R},\ell^2(\mathbb{C}^n))$, if and only if $X$ is absolutely continuous (i.e $X\in C^a(\mathbb{R},\ell^2(\mathbb{C}^n))$ and fulfills for any $k$ the following condition: 
\begin{equation}\dot X_k(t)=\dot X_0(t)e^{-j\omega k t} \ a.e.\label{dphasor2}\end{equation}
\end{theorem}

\begin{proof}
Let us prove the necessity part of this Theorem.
For any $x\in L_{loc}^{2}(\mathbb{R},\mathbb{C}^n)$, one can define its sliding Fourier decomposition $X=\mathcal{F}_T(x) \in L_{loc}^{\infty}(\mathbb{R},\ell^2)$
with phasors: 
$$X_k(t)=\frac{1}{T}\int_{t-T}^t x(\tau)e^{-j\omega k \tau}d\tau.$$

Obviously, $X_k$ is absolutely continuous (defined by an integral) and therefore differentiable almost everywhere: 
\begin{align}
\dot X_k(t)&=\frac{1}{T}(x(t)e^{-j\omega k t}-x(t-T)e^{-j\omega k (t-T)})\nonumber\\
&=\frac{1}{T}(x(t)-x(t-T))e^{-j\omega k t} \ a.e. \label{dphasor}
\end{align}
In particular, for $k=0$,
\begin{equation*}\dot X_0(t)=\frac{1}{T}(x(t)-x(t-T)) \ a.e. 
\end{equation*} 
As a result, the relationship between phasors follows:
\begin{align}
\dot X_k(t)&=\dot X_0(t)e^{-j\omega k t} \ a.e. \label{dphasor0}
\end{align}
Let us now prove the sufficiency part.
Let $X\in L_{loc}^{\infty}(\mathbb{R},\ell^2)$ be an absolutely continuous function and such that the relation \eqref{dphasor0} is satisfied.

For any $t$ fixed, by application of Riesz-Fischer's Theorem (Theorem~\ref{riesz}), there is a function $\bar x\in L^2([t-T,t],\mathbb{C}^n)$ such that, for $\tau\in [t-T,t]$,
$$\bar x(\tau,t)=\sum_{k=-\infty}^{+\infty} X_k(t)e^{j\omega k \tau} ,$$ and 
\begin{equation}X_k(t)=\frac{1}{T}\int_{t-T}^t \bar x(\tau,t)e^{-j\omega k \tau}d\tau.\label{co}\end{equation}
As $\bar x(\tau,t)$ is $T-$periodic in its first argument and defined for any $t$, $\bar x$ can be defined on $\mathbb{R}\times \mathbb{R}$:
\begin{align*}
\bar x:&\ \mathbb{R}\times \mathbb{R}\rightarrow \mathbb{C}^n\\
&(\tau,t)\mapsto \bar x(\tau,t)=\sum_{k=-\infty}^{+\infty} X_k(t)e^{j\omega k \tau}.
\end{align*}
Since the $X_k$ are continuous and differentiable almost everywhere, the partial derivative $\frac{\partial}{\partial t} \bar x(\tau,t)$ in the sense of distributions\footnote{Recall that the derivative of a distribution defined by a series is the distribution defined by the sum of the derivative (in the sense of distributions) of each term of the series.} leads to the expressions: 
\begin{align*}
\frac{\partial}{\partial t} \bar x(\tau,t)&=\sum_{k=-\infty}^{+\infty} \dot X_k(t)e^{j\omega k \tau} \\
&=\dot X_0(t)\sum_{k=-\infty}^{+\infty} e^{j\omega k (\tau-t)} \ a.e. 
\end{align*}
Remember that a Dirac comb, $\Sha_T(t)=\sum_{k=-\infty}^{+\infty} \delta_{kT} (t)$ is a $T-$periodic distribution whose Fourier series is given by $\Sha_T(t)=\frac{1}{T}\sum_{k=-\infty}^{+\infty} e^{j\omega k t}$\cite{Schwartz}. 
This Fourier series does not converge in the classical sense but in the sense of distributions and it converges towards a Dirac comb.
Thus, one can write:
\begin{align}
\frac{\partial}{\partial t} \bar x(\tau,t)&=T\dot X_0(t)\sum_{k=-\infty}^{+\infty} \delta_{t+kT} (\tau) \label{jump}\\
&=T\dot X_0(t)\Sha_T(\tau-t) \ a.e. . \nonumber
\end{align}
The restriction of this distribution to the set $\{(\tau,t)\in \mathbb{R}^2: t-T<\tau<t\}$, leads to:
$$\frac{\partial}{\partial t} \bar x(\tau,t)=0 \ a.e.$$

Therefore, for any $h$ such that $\tau<t+h<\tau+T$, we have the relationship (see Fig. \ref{f3}),
 $\bar x(\tau,t)=\bar x(\tau,t+h).$

 \begin{figure}[h]
\begin{center}
\includegraphics[width=7cm,height=4cm]{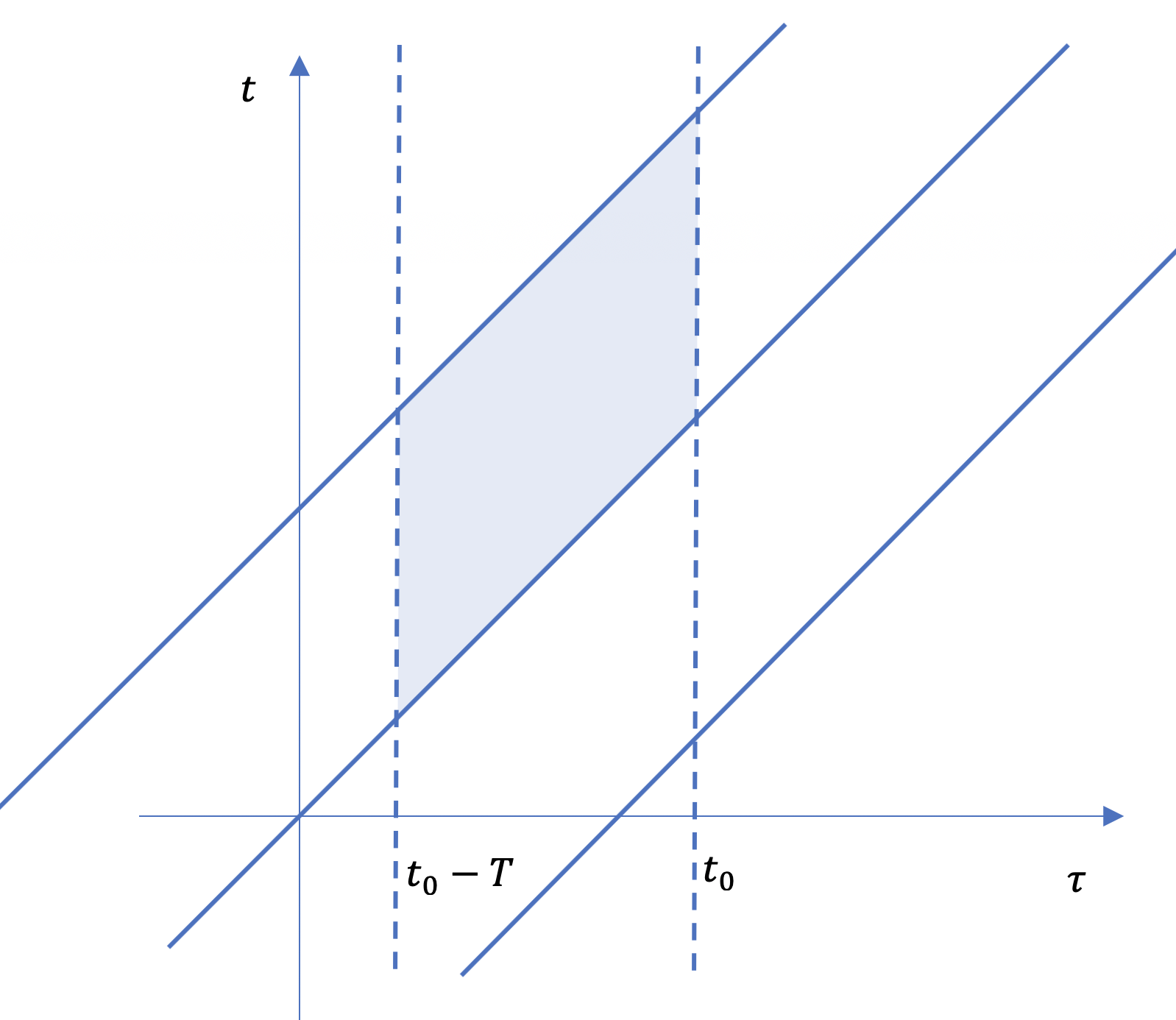}
\caption{The blue parallelogram represents the domain where $\bar x(\tau,t_0)=\bar x(\tau,t_0+h)$ for $t_0-T<\tau<t_0$.}\label{f3}
\end{center}
\end{figure}

Considering the lines parameterized by $h$, $\tau\mapsto t=\tau+h$ with $0<h<T$, strictly included in the set $\{(\tau,t)\in \mathbb{R}^2: t-T<\tau<t\}$, we can define, by overlapping and extension (see Fig. \ref{f5}), a function $x$ defined on $\mathbb{R}$ by: 
\begin{align}
x(\tau)=\bar x(\tau,\tau+h)=\sum_{k=-\infty}^{+\infty} X_k(\tau+h)e^{j\omega k \tau},\label{noncausal}
\end{align}
and such that $X_k(\tau)=\frac{1}{T}\int_{\tau-T}^\tau x(u)e^{-j\omega k u}du.$
This last relation is obtained from equation~\eqref{co} and since $x(u)=\bar x(u,u+h)=\bar x(u,\tau)$ for any $u$ such that $\tau-T<u<\tau$.

Moreover, notice that for all $h$, $0<h<T$, the function $x$ is uniquely defined and independent of the choice of $h$ as $\bar x(\tau,\tau+h_1)=\bar x(\tau,\tau+h_2)$ for all $h_1$ and $h_2$ such that $0<(h_1,h_2)<T$.
Consequently, we have proved by construction that there exists a function $x\in L^2_{loc}$ such that:
$X=\mathcal{F}_T(x).$ 
\end{proof}
\begin{figure}
\begin{center}
\includegraphics[width=7cm,height=4.5cm]{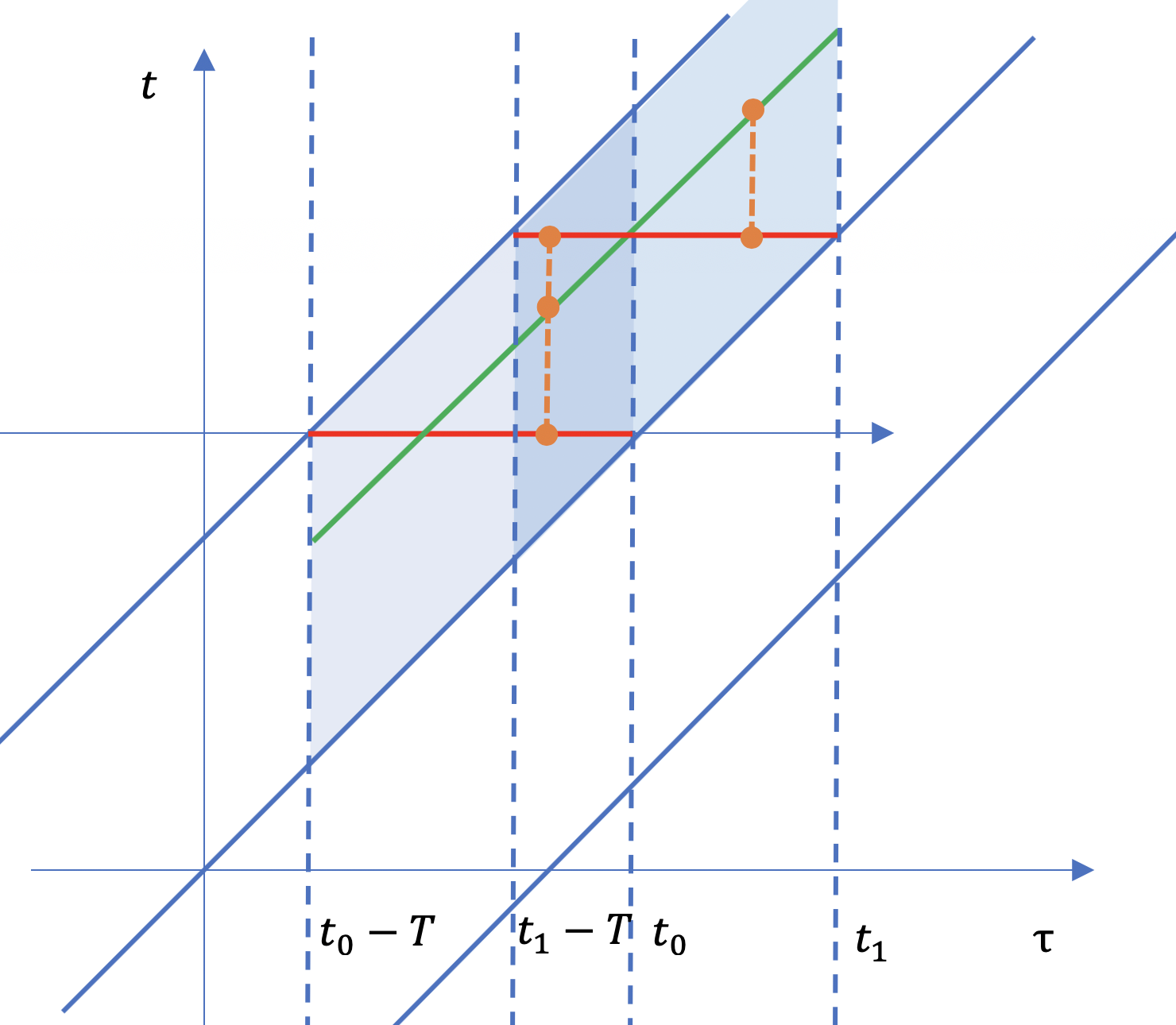}
\caption{Construction of $x(\tau)=\bar x(\tau,\tau+h)$ (green segment) by overlapping and extension. The values of $\bar x$ on the red segments correspond two by two to the values of $\bar x$ on the green segment, therefore to the ones of $x$. }\label{f5}
\end{center}
\end{figure}

\begin{remark}	In the previous proof, from equations~\eqref{dphasor0} and \eqref{jump}, one can notice that for a given $\tau$, the function $\bar x(\tau,\cdot)$ appears as a stair function with jumps $\sigma$ at time $t=\tau+kT$, $k\in\mathbb{Z}$ given by: $$\sigma(\tau+kT)=T\dot X_0(\tau+kT)=x(\tau+kT)-x(\tau+(k-1)T)\ a.e.$$ 
\end{remark}

Let us now show the impact of this result on the punctual convergence of the sliding Fourier series of a piecewise continuous function. The following proposition establishes the appropriate formulas to be used to reconstruct a signal from its sliding Fourier decomposition.

\begin{proposition}[punctual convergence]\label{rec}If $x\in C_{pw}^1$(or $C_{pw}^0$ with bounded variations), then the following reconstruction formulas are equivalent:
\begin{align}
x(t)&=2\sum_{p=-\infty}^{+\infty} X_p(t)e^{j\omega p t}-x(t-T),\label{form1}\\
x(t)&=\sum_{p=-\infty}^{+\infty} X_p(t+h)e^{j\omega p t},\ 0<h<T\label{form3}\\
x(t)&=\sum_{p=-\infty}^{+\infty} X_p(t)e^{j\omega p t}+\frac{T}{2}\dot X_0(t), \label{form2}\end{align} except at points of discontinuity of $x$ for which left and right limits exist.
In addition, if $x\in C^0$, the equalities \eqref{form1}, \eqref{form3} and \eqref{form2} hold everywhere.
\end{proposition}
\begin{proof}
If $x\in C_{pw}^1$(or $C_{pw}^0$ with bounded variations), by construction of $X$, Jordan-Dirichlet's Theorem allows to write for any $t$ and every $\tau$ such that $t-T< \tau< t$:
\begin{equation}\frac{x(\tau^-)+ x(\tau^+)}{2}=\sum_{p=-\infty}^{+\infty} X_p(t)e^{j\omega p \tau} ,\label{nc2}\end{equation}
and for $\tau = t$ (using the periodicity of the series): 
\begin{equation}
\frac{x(t^-)+x(t^+-T)}{2}=\sum_{p=-\infty}^{+\infty} X_p(t)e^{j\omega p t}.\label{s1}\end{equation}
Applying the change of variables $(\tau,t)\mapsto (t,t+h)$ with $0<h<T$, we deduce from \eqref{nc2}:
\begin{equation}\frac{x(t^-)+ x(t^+)}{2}=\sum_{p=-\infty}^{+\infty} X_p(t+h)e^{j\omega p t}.\label{nc3}\end{equation}
Equations \eqref{form1} and \eqref{form3} are then deduced from \eqref{s1} and \eqref{nc3} at points of continuity of $x$.
Equation \eqref{s1} is illustrated in Figure~\ref{f6} for a continuous $x$. 
\begin{figure}[h]
\begin{center}
\includegraphics[scale=.22]{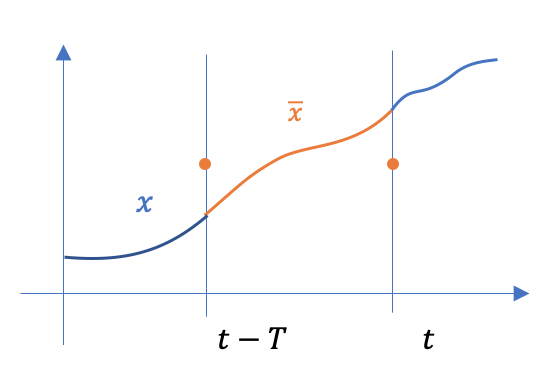}
\caption{The series $\bar x(\tau,t)=\sum_{p=-\infty}^{+\infty} X_p(t)e^{j\omega p \tau}$ for $t-T\leq\tau\leq t$ is plotted in orange and the original signal $x$ in blue. As the series is discontinuous at time $\tau=t$ and takes the value $\frac{x(t^-)+x(t^+-T)}{2}$, it can be observed that clearly, $\bar x(\tau,t)\neq x(t)$ for $\tau=t$.}\label{f6}
\end{center}
\end{figure}
Furthermore, the time derivative of $X_0(t)=\frac{1}{T}\int^t_{t-T}x(\tau)d\tau$ leads to:
\begin{equation}x(t)=T\dot X_0(t)+x(t-T), \label{s2}\end{equation}
everywhere except for a countable number of isolated points for which the right and left limits are well defined. 
Therefore, adding the equations \eqref{s1} and \eqref{s2} leads to formula~\eqref{form2}.
To end the proof, it is obvious that \eqref{form1}, \eqref{form3} and \eqref{form2} hold everywhere if $x\in C^0$.
\end{proof}

It can be noticed that equation \eqref{form2} has the advantage of being causal and it does not imply any delay contrary to equations \eqref{form1} and \eqref{form3}.
Note also that the literature does not mention formulas \eqref{form1} or \eqref{form2} but only formula \eqref{form3} without the strict restriction $0<h<T$. Thus, if $h$ is taken to be equal to zero, the discontinuity of the series makes equation \eqref{form3} erroneous and it must be replaced by formulas \eqref{form1} or \eqref{form2} that take into account the jumps. Otherwise, the Fourier series matches exactly the signal $x$ only if the original signal $x$ is $T$-periodic. One may expect a small gap between $x(t)$ and $x(t-T)$ if $T$ is small enough but this is obviously not the case in general.

Before exploiting the formulas introduced in this section, we introduce the sub-space $H$.
\begin{definition}
The sub-space $H$ of $L^{\infty}_{loc}(\mathbb{R},\ell^2(\mathbb{C}^n))$ is defined as:
\begin{align*}
H=\{X\in C^a(\mathbb{R},\ell^2(\mathbb{C}^n)):&\ \forall k\in\mathbb{Z}, \\
 &\dot X_k(t)=\dot X_0(t)e^{-j\omega k t}\ a.e.\}.
\end{align*}

\end{definition}
The following proposition states that for a.e. $t$, the derivative $\dot X(t)$ of any $X\in H$, is not square summable but only essentially bounded.

\begin{proposition}\label{derive}
If $X\in H$ then for almost all $t$, $\dot X(t)\notin \ell^p$, for $1\leq p<+\infty$ and $\dot X\in L^{2}_{loc}(\mathbb{R},\ell^{\infty})$. Moreover, for any $k$, $\dot X_k$ belongs to $ L_{loc}^{2}(\mathbb{R},\mathbb{C}^n)$.
\end{proposition}

\begin{proof}
As $\forall k, \dot X_k(t)=\dot X_0(t)e^{-j\omega k t}\ a.e.$, one has $|\dot X_k(t)|=|\dot X_0(t)|\ a.e.$ which proves that $\dot X(t)\notin \ell^p$, for $1\leq p<+\infty$.
On the other hand, since the existence of $x\in L^2_{loc}$ such that $X=\mathcal{F}_T(x)$ leads to $\dot X_0(t)=\frac{1}{T}(x(t)-x(t-T))\ a.e.$, we can conclude that $\dot X_0$ belongs to $ L_{loc}^{2}(\mathbb{R},\mathbb{C}^n)$ as also for $\dot X_k$. We deduce that $\dot X\in L^{2}_{loc}(\mathbb{R},\ell^{\infty})$ since for any $a$ and $b$, $$\int_a^b |\dot X(t)|^2_{\ell^{\infty}}dt=\int_a^b |\dot X_0(t)|^2dt<+\infty.$$
\end{proof}
We are now in position to define the inverse of the sliding Fourier decomposition and use the formulas introduced in this section.

\begin{definition}The inverse of the Sliding Fourier Decomposition $\mathcal{F}_T$ is defined by the application $\mathcal{F}_T^{-1}$:
\begin{align*}
\mathcal{F}_T^{-1}: \ &H \rightarrow L^{2}_{loc}(\mathbb{R},\mathbb{C}^n) \\
&X\mapsto x \end{align*}
where $x$ is defined at any time $t$ by the non causal formula:
\begin{align}
x(t)&=\mathcal{F}_T^{-1}(X)(t)=\sum_{p=-\infty}^{+\infty} X_p(t+h)e^{j\omega p t}, \label{ver1}
\end{align}
for any $0<h<T$,
or alternatively by the causal formula, when $X$ has a representative of class $C^1_{pw}$(or $C_{pw}^0$ with bounded variations):
\begin{align}
x(t)&=\mathcal{F}_T^{-1}(X)(t)=\sum_{p=-\infty}^{+\infty} X_p(t)e^{j\omega p t}+\frac{T}{2}\dot X_0(t).\label{ver2}
\end{align}
 \end{definition}

\begin{proposition}\label{bij} Let $X\in L^{\infty}_{loc}(\mathbb{R},\ell^2(\mathbb{C}^n))$.
$X\in H$ if and only if $X=\mathcal{F}_T(\mathcal{F}_T^{-1}(X)).$
\end{proposition}

\begin{proof}
To prove the sufficiency part, let $x=\mathcal{F}_T^{-1}(X)$ given by \eqref{ver1} or \eqref{ver2}. 
If $X=\mathcal{F}_T(\mathcal{F}_T^{-1}(X))$ then $X=\mathcal{F}_T(x)$. As for a.e. $t$, $X(t)\in \ell^2$, then $x\in L^2({[t-T\ t]})$ (by Riesz Fisher Theorem) and thus $x\in L^2_{loc}$
and it can be concluded that $x$ is (by definition) a representative of $X$. 


To prove the necessity part, notice that if $X$ belongs to $H$, then $X$ has a representative $x\in L^2_{loc}$ as it has been shown in the proof of Theorem~\ref{coincidence}. It is provided by the non causal formula:
 \begin{align*}
x(t)&=\sum_{k=-\infty}^{+\infty} X_k(t+h)e^{j\omega k t} \ a.e.,
\end{align*}
for any $0<h<T$.
The result follows since by construction in the proof of Theorem~\ref{coincidence}, it has been proved that $X=\mathcal{F}_T(x)$.
In addition, if $x$ is $C^1_{pw}$(or $C_{pw}^0$ with bounded variations), formula \eqref{ver1} is equivalent to \eqref{ver2} as stated by Proposition~\ref{rec}. 
\end{proof}
\begin{remark}
For a given function $x \in L^2_{loc}$, if we define $v$ by $v=\mathcal{F}_T^{-1}(\mathcal{F}_T(x))$ using the non causal formula \eqref{ver1}, Riesz-Fisher's Theorem allows to conclude that the series $v$ converges to $x$ in the $L^2([t-T,t],\mathbb{C}^n)$ sense. Without more assumptions on the regularity of $x$, it is not obvious to prove that the causal formula \eqref{ver2} leads to the same result. 

\end{remark}

\section{Harmonic modeling}

\subsection{General case: Carathéodory systems}

We consider nonlinear dynamical systems described by:
\begin{equation}
\dot x=f(t,x) \label{dg}, \quad x(0) = x_0
\end{equation}
with the following assumption.

\begin{assumption} \label{as1}The differential equation \eqref{dg} admits solutions in the Carathéodory sense\footnote{see e.g. Chapter 19, p.p. 530-533 in \cite{Poznyak}} and for any $x\in L^2_{loc}(\mathbb{R},\mathbb{C}^n)$, the time function $f(t,x(t))$ belongs to $L^1_{loc}(\mathbb{R},\mathbb{C}^n)$. 
\end{assumption}

\begin{theorem}\label{CNS} Under Assumption \ref{as1}, if $x$ is a solution of the differential equation \eqref{dg} in the Carathéodory sense, then, for any $T>0$, $X=\mathcal{F}_T(x)$ is a solution of:
\begin{equation}
		\dot X=\mathcal{F}_T(f(t,x))(t)-\mathcal{N}X, \quad
		X(0)=\mathcal{F}_T(x)(0)
\label{hdg}
\end{equation}
with $\mathcal{N}=diag(j\omega k\ Id_n,\ k\in \mathbb{Z})$. Reciprocally, if $X\in H$ is a solution of \eqref{hdg}, then its representative $x$ (i.e. $X=\mathcal{F}_T(x))$ is a solution of \eqref{dg} and $x=\mathcal{F}_T^{-1}(X)$.
In addition, for any $k\in\mathbb{Z}$, the phasors $X_k \in C^1(\mathbb{R},\mathbb{C}^n)$ and $\dot X\in C^0(\mathbb{R},\ell^{\infty}(\mathbb{C}^n))$.
Moreover, if the solution $x$ is unique for the initial condition $x_0$, then $X$ is unique for the initial condition $X(0)=\mathcal{F}_T(x)(0)$. The reciprocal is not generally true.

\end{theorem}

\begin{proof}
Let us assume that $x$ is a solution of \eqref{dg} in the sense of Carathéodory and defined on $\mathbb{R}$. 
By setting $v=\tau-t$, for every fixed $t$, the phasors are given by:
\begin{align*}
X_k(t)&=\frac{1}{T}\int_{t-T}^t x(\tau)e^{-j\omega k \tau}d\tau\\
&=\frac{1}{T}\int_{-T}^0 x(v+t)e^{-j\omega k (v+t)}dv.
\end{align*}
Since the function $x$ is absolutely continuous, it is therefore differentiable $a.e$, and bounded over any bounded interval. Hence : 
\begin{align*}
\frac{dx(v+t)e^{-j\omega k (v+t)}}{dt}=\frac{dx(v+t)}{dt}&e^{-j\omega k (v+t)}-\cdots\\
&j\omega k\ x(v+t)e^{-j\omega k (v+t)},\end{align*}
for almost all $v\in[-T,0]$ and there is a constant $C$ such that, for any $k$ and for all $v\in[-T,0]$ :
 $$|-j\omega k \ x(v+t)e^{-j\omega k (v+t)}|=\omega k |x(v+t)|\leq C.$$ 

Always for fixed $t$, as $x(v+t)$ is continuous for $v\in[-T,0]$, $x(v+t)\in L^2{[-T\ 0]}$ and by assumption~\ref{as1}, it follows that $g(v+t)=f(v+t,x(v+t))\in L^1{([-T\ 0])}$.
As a result:
$$\left|\frac{dx(v+t)e^{-j\omega k (v+t)}}{dt}\right|\leq |g(v+t)|+C,$$
for almost all $v\in[-T,0]$, which proves that the term $|\frac{dx(v+t)e^{-j\omega k (v+t)}}{dt}|$ is bounded by an integrable function on $[-T\ 0]$. As a consequence, a differentiation under the integral sign can be applied and we have:
\begin{align}
\dot X_k(t)&=\frac{1}{T}\int_{-T}^0 \frac{dx(v+t)e^{-j\omega k (v+t)}}{dt}dv\nonumber\\
&=\frac{1}{T}\int_{-T}^0 \frac{dx(v+t)}{dt}e^{-j\omega k (v+t)}dv-\cdots\nonumber\\
&\qquad \quad j\omega k\frac{1}{T}\int_{-T}^0 x(v+t)e^{-j\omega k (v+t)}dv\nonumber\\
&=<f(x,t)>_k(t)-j\omega k X_k(t)\ a.e..\label{p1}
\end{align}
where the notation $<v>_k(t)$ refers to the $k$-th phasor of $v$ at time $t$: $<v>_k(t)=\frac{1}{T}\int_{t-T}^tv(\tau)e^{-j\omega k \tau}d\tau$.

Let: \begin{align}\mathcal{N}&=diag(j\omega k\ Id_n, k\in \mathbb{Z})\label{N}\\
&=\left[
\begin{array}{ccccc}
\ddots & & \vdots & &\udots \\ & -j\omega\ Id_n & 0 & 0 & \\
\cdots & 0& 0 & 0 & \cdots \\
 & 0 & 0 & j\omega\ Id_n& \\
\udots & & \vdots & & \ddots\end{array}\right],\nonumber\end{align} we have:
$$\dot X=\mathcal{F}_T(f(\tau,x(\tau))(t)-\mathcal{N}X\ a.e.,$$
which establishes the dynamic phasor model. 

In addition, $X\in H$ by Theorem~\ref{coincidence} and as $x$ is continuous, for any $k\in\mathbb{Z}$, the phasor 
$X_k(t)=\frac{1}{T}\int_{t-T}^t x(\tau)e^{-j\omega k \tau}d\tau$ 
belongs to $C^1(\mathbb{R},\mathbb{C}^n)$ and the relation \eqref{dphasor2} is satisfied everywhere.
Following similar steps as in the proof of Proposition~\ref{derive}, we have $\dot X \in C^0(\mathbb{R},\ell^{\infty}(\mathbb{C}^n))$.

Conversely, assume that $X \in H$ is a solution of the differential equation:
\begin{equation} \dot X_k(t)=<f(t,x)>_k(t)-j\omega k X_k(t),\ k\in\mathbb{Z}\label{phasor}\end{equation} 
where $x$ denotes a representative of $X$.
First at all, it can be observed that $X$ is necessarily a Carathéodory solution. Indeed, if $X\in H$, then $X$ is an absolutely continuous function and consequently it satisfies the integral equation 
\begin{align*}
X(t)&=X(0)+\int_0^t\mathcal{F}_T(f(\nu,x(\nu))(\tau)-\mathcal{N}X(\tau)d\tau.
 \end{align*}
 

Now, as $X$ admits a representative $x \in L^2_{loc}(\mathbb{R},\mathbb{C}^n)$ such that $X=\mathcal{F}_T(x)$, we have to show that this representative $x$ satisfies the differential equation \eqref{dg} almost everywhere. To this end, let us consider equation~\eqref{phasor} and notice that the phasors: $<f(x,t)>_k(t)$ for any $k$, are well defined for such a $x$ by Assumption \ref{as1}. Let $z$ be a primitive of $f(t,x)$ ($z$ is well defined by Assumption~\ref{as1}) and integrating by parts (of absolutely continuous and a.e. differentiable functions), one gets:
{\small
\begin{align*}
\dot X_k(t)&=<f(t,x)>_k(t)-j\omega k X_k(t)\\
&= \frac{1}{T}\int_{t-T}^t f(\tau,x(\tau))e^{-j\omega k \tau}d\tau-\frac{1}{T}\int_{t-T}^t x(\tau)j\omega ke^{-j\omega k \tau}d\tau\\
&= \frac{1}{T}(z(t)-z(t-T))e^{-j\omega k t}+\dots\\&\frac{1}{T}\int_{t-T}^t z(\tau)j\omega ke^{-j\omega k \tau}d\tau-
\frac{1}{T}\int_{t-T}^t x(\tau)j\omega ke^{-j\omega k \tau}d\tau.
\end{align*}}
As $X\in H$, we have for any $k$:
\begin{equation}\dot X_k(t)=\dot X_0 e^{-j\omega k (t)}=\frac{1}{T}(x(t)-x(t-T))e^{-j\omega k t} \ a.e.\end{equation}
By identification, we obtain, a.e.: 
\begin{align}\frac{1}{T}&(x(t)-x(t-T))e^{-j\omega k t}=\frac{1}{T}(z(t)-z(t-T))e^{-j\omega k t}+\dots\nonumber\\&\frac{1}{T}\int_{t-T}^t z(\tau)j\omega ke^{-j\omega k \tau}d\tau-\frac{1}{T}\int_{t-T}^t x(\tau)j\omega ke^{-j\omega k \tau}d\tau.\label{ipp}
\end{align}
For $k=0$, we have:
$$\frac{1}{T}(x(t)-x(t-T))=\frac{1}{T}(z(t)-z(t-T))\ a.e.$$
Hence, for any $k$, the expression \eqref{ipp} reduces to:
$$0=\frac{1}{T}\int_{t-T}^t (z(\tau)-x(\tau))\ j\omega k\ e^{-j\omega k \tau}d\tau\ a.e.,$$
Dividing these expressions by $ j\omega k$, for $k\neq0$, by unicity of the Fourier decomposition 
we conclude that $z$ and $x$ differ only by a constant $z_0$, that is: $z=x+z_0 \ a.e.$
As $z$ is a primitive of $f(t,x)$, we conclude that $x$ is differentiable $a.e.$ and satisfies the differential equation \eqref{dg}. Moreover, $x$ is absolutely continuous and by Proposition~\ref{rec} we have $x(t)=\mathcal{F}_T^{-1}(X)(t)$ for all $t$.
In particular, $x(0)= \mathcal{F}_T^{-1}(X)(0)$.

Finally, as in the first part of the proof, the continuity of $x$ implies that for any $k\in\mathbb{Z}$, $X_k$ is $C^1(\mathbb{R},\mathbb{C}^n)$
and the relation \eqref{dphasor2} is satisfied everywhere.

Concerning the uniqueness part, if \eqref{dg} admits an unique solution for a given $x_0$, then the unicity of the trajectory $x$ leads to the unicity of the initial condition for \eqref{hdg}.
If \eqref{hdg} admits several solutions in $H$, then for two distinct solutions $X_i$, $i=1, 2$, one would associate distinct solutions $x_i(t)=\mathcal{F}_T^{-1}(X_i)(t)$, $i=1,2$ with $x_0= \mathcal{F}_T^{-1}(X_i)(0)$, $i=1,2$ which would contradict the hypothesis of uniqueness of the $x$. On the other hand, it is not possible to conclude that there is a unique solution to \eqref{dg} with 
$ x(0)=x_0=\mathcal{F}_T^{-1}(X)(0)$ from the uniqueness of the solution of \eqref{hdg} with $X(0)=X^0$.
Indeed, if there are two solutions $x_1$ and $x_2$ to \eqref{dg} with 
$ x(0)=x_0$, the initial conditions generated by $x_1$ and $x_2$: 
$$X_{1}(0)=\mathcal{F}_T(x_1)(0), \quad X_{2}(0)=\mathcal{F}_T(x_2)(0),$$
are distinct if $x_1$ and $x_2$ do not coincide on $[-T\ 0]$ and each initial condition $X_1(0)$ and $X_2(0)$ leads to two distinct problems in the harmonic domain.
\end{proof}

To have $X \in H$, the initial conditions for the harmonic system are constrained by $X(0)=\mathcal{F}_T(x)(0)$ with $x$ solution of \eqref{dg}. This is obvious, as by Riesz-Fisher Theorem, any $X(0)\in \ell^2$ leads to $x\in L^2([-T\ 0])$ which is not a priori a solution of \eqref{dg}. Nevertheless, for the equilibria of the harmonic model, one has the following useful result.

\begin{corollary}[Equilibria Admissibility]\label{equi}$x$ is a $T$-periodic trajectory of $\dot x=f(x,t)$ if and only if $X=\mathcal{F}_T(x)$ is an equilibrium of $\dot X=\mathcal{F}_T(f(\tau,x(\tau))(t)-\mathcal{N}X$.
\end{corollary}

\begin{proof}
If $x$ is a $T$-periodic trajectory of $\dot x=f(x,t)$, then $X=\mathcal{F}_T(x)$ is constant and it follows that $\dot X=0$. Thus $\mathcal{F}_T(f(\tau,x(\tau))(t)-\mathcal{N}X$ is equal to zero at any time instant and defines an equilibrium $X$ of the harmonic equation. Conversely, if $X$ is an equilibrium, then as $X$ is constant, it is trivially in $H$ since it defines an absolutely continuous function whose derivative satisfies $\dot X_k=0$ for all $k$. Consequently there exists a $T-$periodic representative $x=\mathcal{F}_T^{-1}(X)$ that satisfies the differential equation $\dot x=f(x,t)$.
\end{proof}

%
%
%
%

\subsection{LTV (LTP, bilinear affine) systems case}\label{sysaffine}
In this section, we explain how Theorem~\ref{CNS} can be applied to different classes of LTV and bilinear affine systems. Recall that the sliding Fourier decomposition of a product $A(t)x(t)$ is determined as a convolution product (see Property~\ref{product} in Appendix): 
$$\mathcal{F}_T(Ax)=\mathcal{T}_T(A)X=\mathcal{A}X,$$
where $X=\mathcal{F}_T(x)$ and where the block Toeplitz transformation $\mathcal{T}_T(A)$ defines an infinite dimensional matrix function as follow: 
\begin{align*}\mathcal{A}(t)&=\mathcal{T}_T(A)(t)\\&=
\left[
\begin{array}{ccccc}
\ddots & & \vdots & &\udots \\ & A_0(t) & A_{-1}(t) & A_{-2}(t) & \\
\cdots & A_{1}(t) & A_0(t) & A_{-1}(t) & \cdots \\
 & A_{2}(t) & A_{1} (t)& A_0(t) & \\
\udots & & \vdots & & \ddots\end{array}\right],\end{align*}
with $A_k(t)=\frac{1}{T}\int_{t-T}^t A(\tau)e^{-j\omega k \tau}d\tau$.

The following proposition treats the cases of LTV and LTI systems. 
\begin{proposition}\label{lineaire}
Assume that $A$ and $B$ $\in L^{2}_{loc}$ and $u\in L^{2}_{loc}$.
Then, $x$ is the unique solution of the LTV system.
\begin{align}
\dot x=A(t)x(t)+B(t)u(t), \quad x(0)=x_0 \label{ltp}
\end{align} 
if and only if $X\in H$ is the unique solution of the LTV system
\begin{align*}
\dot X(t)=(\mathcal{A}(t)-\mathcal{N})X(t)+\mathcal{B}(t)U(t), \quad X(0)=\mathcal{F}_T(x)(0)
\end{align*}
where $X=\mathcal{F}_T(x)$, $U=\mathcal{F}_T(u)$, $\mathcal{A}=\mathcal{T}_T(A)$ and $\mathcal{B}=\mathcal{T}_T(B)$ . \\
If, in addition, $A$ and $B$ are $T$-periodic then $\mathcal{A}$ and $\mathcal{B}$ are constant and the harmonic system is LTI.
\end{proposition}
\begin{proof}
If $A(t),\ B(t)$ and $u(t)\in L^2_{loc}$ then, using generalized Hölder's inequality\footnote{Let $p,q,r\in [1,+\infty]$ be such that $\frac{1}{p} + \frac{1}{q} = \frac{1}{r}$. 
Let $f\in L^p$ and $g\in L^q$. Then, $f g\in L^r$ and $\|fg\|_{L^r}\leq \|f\|_{L^p}\|g\|_{L^q}.$}, Assumption~\ref{as1} is satisfied as well as the uniqueness condition (see Chapter 19, p.p. 530-533 in \cite{Poznyak}). Thus, the result follows by Theorem~\ref{CNS} and by Property~\ref{product} given in the appendix which provides the harmonic model. 
\end{proof}


In the sequel, we discuss the harmonic modeling of bilinear affine systems. This class of systems is of particular interest as the analysis can be extended to the case of affine switched systems and applied to power converters (see e.g. \cite{Beneux}). Let:
\begin{equation}
\dot{x}(t)= A(s(t))x(t) + B(s(t))w(t), \quad x(0)=x_0\label{affine}
\end{equation}
where $s(\cdot)$ is the control variable and $w(\cdot)$ an exogenous input.
It is assumed that the control dependent matrices can be decomposed into an affine function with respect to $s(t)$:
\begin{equation*} 
\begin{array}{rcl}
A(s(t)) &=& A_{ind}+s(t)A_{dep}, \\
B(s(t)) &=& B_{ind}+s(t)B_{dep},
\end{array}
\end{equation*}
where $s(t)$ is a bounded scalar and $A_{ind},B_{ind},A_{dep},B_{dep}$ are constant matrices. 

\begin{proposition}\label{model_harmo_affine}
Assume that $s\in L^{2}_{loc}$ and $w\in L^{2}_{loc}$.
Then, $x$ is the unique solution of \eqref{affine} if and only if $X$ is the unique solution for
\begin{align}
\dot{X}(t) &= (\mathcal{A} (S(t))-\mathcal{N}) X(t)+\mathcal{B} (S(t)) W(t)\label{affineh}\\
 X(0)&=\mathcal{F}_T(x)(0),\nonumber
\end{align}
with
\begin{align}
\mathcal{A} (S(t))&=\mathcal{A}_{ind}+\mathcal{S}(t)\otimes A_{dep}, \label{r1}\\
\mathcal{B} (S(t))&=\mathcal{B}_{ind}+ \mathcal{S}(t)\otimes B_{dep}, \label{r2}
\end{align}
where $\otimes$ represents the Kronecker product, $X=\mathcal{F}_T(x)$, $W=\mathcal{F}_T(w)$, $\mathcal{S}=\mathcal{T}_T(s)$, $\mathcal{A}_{ind}=\mathcal{T}_T(A_{ind})$ and $\mathcal{B}_{ind}=\mathcal{T}_T(B_{ind})$. 
\end{proposition}
\begin{proof}
As $s\in L^{2}_{loc}$, it is clear that $A(s)$ and $B(s)$ are $L^{2}_{loc}$. Consequently, the bilinear affine system \eqref{affine} appears as a special case of LTV systems and Proposition~\ref{lineaire} applies. Property~\ref{product} given in the appendix leads to \eqref{r1} and \eqref{r2}. 
\end{proof}

\section{Harmonic modelling based stability and stabilization}
In this section, the impact on stability analysis and stabilization using harmonic modeling of the results proposed previously is explcited. 
\subsection{Stability of LTP systems}\label{stab}
Here, we consider linear periodic systems and provide stability results that generalize existing ones. Indeed, we do not assume like in \cite{Zhou2008} that the matrix function $A$ is piecewise continuous and differentiable almost everywhere but only of class $L^2$. Moreover, contrary to \cite{Bittanti} and \cite{Zhou2002}, we do not need the Floquet Theorem. Our result in this subsection is a direct consequence of the necessary and sufficient condition given in Theorem~\ref{CNS} and concerns periodic linear systems given by
\begin{align}
	\dot x(t)=A(t)x(t)
	\label{autox}
\end{align} 
and their harmonic model:
\begin{align} \dot X(t)=(\mathcal{A}-\mathcal{N})X(t)
		\label{autoXN}
\end{align}
\begin{theorem}[Harmonic Lyapunov equation]\label{Lyapunov}
Assume that $A\in L^{2}([0 \ T])$ is a $T$-periodic matrix function with real values and let $Q\in L^{\infty}([0\ T])$ be a $T$-periodic symmetric and positive definite matrix function. $P$ is the unique $T$-periodic symmetric positive definite solution of the periodic Lyapunov differential equation:
$$\dot P(t)+A'(t)P(t)+P(t)A(t)+Q(t)=0,$$
if and only if $\mathcal{P}=\mathcal{T}_T(P)$ is the unique hermitian and positive definite solution of the Lyapunov algebraic equation:
\begin{equation}
\mathcal{P}(\mathcal{A}-\mathcal{N})+(\mathcal{A}-\mathcal{N})^*\mathcal{P}+\mathcal{Q}=0,\label{al}
\end{equation}
where $\mathcal{Q}=\mathcal{T}_T(Q)$ is hermitian positive definite and $\mathcal{A}=\mathcal{T}_T(A)$.
Moreover, $\mathcal{P}$ is a bounded operator on $\ell^2$.
\end{theorem}
\begin{proof}
$A\in L^{2}([0 \ T])$ and $Q\in L^{\infty}([0\ T])$ ensure that the Carathéodory conditions and the uniqueness condition are satisfied.
Moreover, by generalized Hölder's inequality, $A'(t)P(t)+P(t)A(t)+Q(t)$ is $L^1_{loc}$ for any $P\in L^2_{loc}$. Hence, Assumption~\ref{as1} is satisfied and the NSC of Theorem~\ref{CNS} applies.
Let us show that the sliding Fourier decomposition taken on a window corresponding to the period $T$ leads to the harmonic equation~\eqref{al}.

For any $T$ and any $k$, the sliding Fourier decomposition allows to write:
\begin{equation}\dot {P}_k=-<A'P+PA+Q>_k(t)-j\omega k P_k.\label{phk}\end{equation}
The Toeplitz form associated with \eqref{phk} can be determined (see Definition~\ref{toeplitz} in the appendix). One gets for the left side of \eqref{phk}: 
\begin{align*}\dot {\mathcal{P}}&=\mathcal{T}_T(\dot {P})=\left[
\begin{array}{ccccc}
\ddots & & \vdots & &\udots \\ & \dot P_0(t) & \dot P_1^*(t) & \dot P_2^*(t) & \\
\cdots & \dot P_{1}(t) & \dot P_0(t) & \dot P_1^*(t) & \cdots \\
 & \dot P_{2}(t) & \dot P_{1} (t)& \dot P_0(t) & \\
\udots & & \vdots & & \ddots\end{array}\right],\end{align*}
which is hermitian since $P(t)=P(t)'$ and $\mathcal{T}_T(P')=\mathcal{P}^*$.
For the first term of the right member of \eqref{phk}, we have (by Property~\ref{product} in the appendix):
\begin{align*}\mathcal{P}\mathcal{A}+\mathcal{A}^*\mathcal{P}+\mathcal{Q}=
\mathcal{T}_T(A'P+PA+Q).\end{align*}
where $\mathcal{P}=\mathcal{T}_T(P)$, $\mathcal{A}=\mathcal{T}_T(A)$ and $\mathcal{Q}=\mathcal{T}_T(Q)$.
Finally, for the last terms of the right member of \eqref{phk}, using the fact that $\mathcal{P}$ is hermitian it is easy to verify that:
\begin{align*}
\mathcal{T}_T(\mathcal{N}P)&=\left[
\begin{array}{ccccc}
\ddots & & \vdots & &\udots \\ 
& 0 & -j\omega {P}_1^*(t)& -j2\omega {P}_2^*(t) & \\
\cdots & j\omega {P}_1(t) &0 & -j\omega {P}_1^*(t) & \cdots \\
 & j2\omega {P}_2(t) & j\omega {P}_1(t)& 0 & \\
\udots & & \vdots & & \ddots
\end{array}\right]\\&=\mathcal{P}\mathcal{N}+\mathcal{N}^*\mathcal{P}.\end{align*}
where $\mathcal{N}$ is given by \eqref{N}. Hence, the final result follows:
$$\dot {\mathcal{P}}+\mathcal{P}(\mathcal{A}-\mathcal{N})+(\mathcal{A}-\mathcal{N})^*\mathcal{P}+\mathcal{Q}=0$$
with the initial condition $P(0)=\mathcal{F}_T(P)(0)$.
Note that $\mathcal{P}$ and $\mathcal{Q}$ are necessarily hermitian and positive definite (see Proposition~\ref{defpos} in the appendix). Now, if the length of the sliding window is equal to the period $T$, as $P$, $A$ and $Q$ are $T$-periodic, we have $\mathcal{P}$, $\mathcal{A}$ and $ \mathcal{Q}$ constants and $\dot {\mathcal{P}}=0$.
We conclude that $\mathcal{P}$ is the unique hermitian positive definite solution of the algebraic equation and we have necessarily $\mathcal{P}(t)=\mathcal{P}(0)=\mathcal{T}_T(P)$. 

Conversely, since the unique solution of the harmonic algebraic equation is independent of the initial condition (clearly in $H$ by Corollary~\ref{equi}), if $P(t)$ is a representative of $\mathcal{P}$, i.e. $\mathcal{P}=\mathcal{T}_T(P),$ it is unique. Moreover, as Theorem~\ref{CNS} states that this representative satisfies the periodic Lyapunov differential equation, $P(t)$ is absolutely continuous. Thus $\mathcal{P}$ is a bounded operator on $\ell^2$ following Property~\ref{borne} given in the appendix.
\end{proof}
\begin{corollary}\label{CNS2}
The linear periodic system \eqref{autox} where $A\in L^{2}([0\ T])$ is a $T$-periodic matrix function, is globally asymptotically stable if and only if for any $T$-periodic symmetric positive definite matrix function $Q \in L^{\infty}([0\ T])$, the hermitian and positive definite matrix $\mathcal{P}$ is the unique solution of the algebraic Lyapunov equation:
$$\mathcal{P}(\mathcal{A}-\mathcal{N})+(\mathcal{A}-\mathcal{N})^*\mathcal{P}+\mathcal{Q}=0,$$
where $\mathcal{Q}=\mathcal{T}_T(Q)$. 
\end{corollary}
\begin{proof} As the existence of a solution to the harmonic Lyapunov algebraic equation is equivalent to the existence of a solution of the periodic Lyapunov differential equation, the result follows.
\end{proof}
\begin{corollary}\label{lyap}
The following are equivalent:
\begin{enumerate}
\item $\mathcal{V}(X)=X^*\mathcal{P}X$ is a Lyapunov function for the harmonic system \eqref{autoXN}. 
\item $v(x)=x'Px$ is a Lyapunov function for the periodic linear system \eqref{autox}.
\item the functional 
\begin{align*}
v(x)(t)&=<x,Px>_{L^2([t-T,t])}(t)\\&=\frac{1}{T}\int_{t-T}^tx(\tau)'P(\tau)x(\tau)d\tau
\end{align*} 
is a Lyapunov function for the periodic linear system \eqref{autox}. 
\end{enumerate}
\end{corollary}
\begin{proof}
The fact that the two first items are equivalent is obvious. 
For the third one, the equivalence between 1 and 3 is shown via Parseval Theorem~\ref{thparseval}. As $\mathcal{P}$ is a bounded operator on $\ell^2$, the function $\mathcal{V}(X)=X^*\mathcal{P}X$ is well defined and differentiable $a.e.$ along a trajectory since $X$ is absolutely continuous. Its derivative is also well defined since $\dot {\mathcal{V}}(X)=-X^*QX$ with $\mathcal{Q}X\in \ell^2$ for all $X\in \ell^2$ (as $\mathcal{Q}$ is a bounded operator; see Property~\ref{borne} in the appendix). Similarly, for any $x\in L^2([t-T,t])$, as $P(t)$ is absolutely continuous, $Px \in L^2([t-T, t])$. Also, the functional $v(x)(t)$ is well defined for any $x\in L^2([t-T,t])$ as well as its derivative since $Qx\in L^2([t-T,t])$ for any $x\in L^2([t-T,t])$ as $Q\in L^{\infty}([0\ T])$.

Using Parseval Theorem~\ref{thparseval}, we get the equality:
 \begin{align*}v(x)(t)&=<x,Px>_{L^2([t-T,t])}(t)\\&=\frac{1}{T}\int_{t-T}^tx(\tau)'P(\tau)x(\tau)d\tau\\&=\mathcal{V}(X)(t), \end{align*}
 and the derivative leads to:
 \begin{align*}
\dot v(x)(t)&=-<x,Qx>_{L^2([t-T,t])}(t)=-X^*\mathcal{Q}X=\dot {\mathcal{V}}(X)(t).
\end{align*}
The result follows 
noticing that: $x=0\ a.e.$ on $[T-t,t]$ if and only if $X=0$.
\end{proof}

\subsection{Stabilization of harmonic systems}
In this section, useful results concerning state feedback design for harmonic systems are provided. We start by a proposition of practical importance for general harmonic control design and which is a consequence of the Coincidence Condition given in Theorem~\ref{coincidence}.

\begin{proposition}[Static gain]\label{static}Let $X\in H$. If $U=\mathcal{K}X$ with $\mathcal{K}$ a static gain then $U\in H$ if and only if $\mathcal{K}$ is a block Toeplitz matrix i.e. $\mathcal{K}=\mathcal{T}_T(K(t))$ where $K(t)$ is a $T$-periodic matrix function of class $L^2_{loc}$. \end{proposition}

\begin{proof}
Let $U=\mathcal{K}X$ with $\mathcal{K}$ a static gain, then the time derivative of the $k$-th phasor is given by:
\begin{align*}
\dot U_k&=\sum_{p=-\infty}^{+\infty}\mathcal{K}_{kp}\dot X_p
\end{align*}
where $\mathcal{K}_{kp}$ are block matrices of $\mathcal{K}$ of appropriate dimensions.
If $U\in H$ then
$\dot U_k(t)=e^{-j\omega k t}\dot U_0(t)\ a.e.,$
and it follows that 
almost everywhere:
\begin{align*}
\dot U_k&=\sum_{p=-\infty}^{+\infty}\mathcal{K}_{kp}\dot X_p=e^{-j\omega k t}\dot U_0(t)\\
\dot U_k&=e^{-j\omega k t}\sum_{p=-\infty}^{+\infty}\mathcal{K}_{0p}\dot X_p=\sum_{p=-\infty}^{+\infty}\mathcal{K}_{0p}\dot X_{p+k},
\end{align*}
since $\dot X_{p+k}=e^{-j\omega k t}\dot X_p$ (The term $e^{-j\omega k t}$ acts as a shift operator on the components of $\dot X$).
In other words, it is necessary that:
$$\mathcal{K}_{kp}=\mathcal{K}_{0(p-k)},$$
and therefore $\mathcal{K}$ is a bloc Toeplitz matrix. The reciprocal is obvious. Notice that $K(t)$ must be $T-$periodic and at least of class $L^2_{loc}$ so that the Fourier decomposition of the $K(t)x(t)$ with $x(t)$ continuous, can be computed and associated with its Fourier decomposition via Riesz-Fischer's Theorem~\ref{riesz}. \end{proof}

The synthesis of a state feedback in the harmonic domain imposes a bloc Toeplitz structure on the gain matrix $\mathcal{K}$. If this is not the case, the relation between harmonics:
$$\dot U_k(t)=e^{-j\omega k t}\dot U_0(t),$$
is not satisfied and $U$ is not the sliding Fourier decomposition of any temporal signal $u$.
A naive harmonic control design that does not respect this mandatory condition, leads to erroneous conclusions on closed loop stability in the time domain, even if the closed loop harmonic model is globally asymptotically stable as we have in this case: $\mathcal{F}_T\mathcal{F}_T^{-1}(U)\neq U.$

The following Theorem provides a state feedback control design method for LTP systems with very weak regularity assumptions. It extends available result \cite{Zhou2008} to more general LTP systems.

\begin{theorem}[Harmonic Riccati equation]
Consider $T$-periodic matrix functions $A\in L^{2}([0 \ T])$, $B\in L^{\infty}([0\ T])$ and $T$-periodic symmetric positive definite matrix functions of $L^{\infty}$ class $R$ and $Q$. Assume that there exists a number $\eta>0$ such that the set $\{t: |det(R(t))|<\eta\}$ is of zero measure. Then, $P$ is the unique $T$-periodic symmetric positive definite solution of the periodic Riccati differential equation:
\begin{align*}\dot P(t)&+A'(t)P(t)+P(t)A(t)\\&-P(t)B(t)R^{-1}(t)B(t)'P(t)+Q(t)=0,\end{align*}
if and only if the matrix $\mathcal{P}=\mathcal{T}_T(P)$ is the unique hermitian and positive definite solution of the algebraic Riccati equation:
$$\mathcal{P}(\mathcal{A}-\mathcal{N})+(\mathcal{A}-\mathcal{N})^*\mathcal{P}-\mathcal{P}\mathcal{B}\mathcal{R}^{-1}(t)\mathcal{B}^*\mathcal{P}+\mathcal{Q}=0,$$
where $\mathcal{Q}=\mathcal{T}_T(Q)$ is hermitian positive definite.
Moreover, $\mathcal{P}$ is a bounded operator on $\ell^2$.
\end{theorem}
\begin{proof}
By Property~\ref{inv} given in the appendix, $R^{-1}$ is defined a.e. and $\mathcal{R}^{-1}=\mathcal{T}_T(R^{-1})$.
Following the assumptions made on matrices $A$, $B$, $Q$ and $R$, it is straightforward to show that the matrix differential equation satisfies Assumption~\ref{as1}.
For any $T>0$, by applying the CNS of Theorem~\ref{CNS} on a window of length $T$ and following similar steps as in the proof of Theorem~\ref{Lyapunov}, we obtain that if $P$ is the unique positive definite solution then $\mathcal{P}$ is the unique solution ($\mathcal{P}\in H$) of the differential equation:
$$\dot {\mathcal{P}}+\mathcal{P}(\mathcal{A}-\mathcal{N})+(\mathcal{A}-\mathcal{N})^*\mathcal{P}-\mathcal{P}\mathcal{B}\mathcal{R}^{-1}\mathcal{B}^*\mathcal{P}+\mathcal{Q}=0,$$
with $P(0)=\mathcal{F}_T(P)(0)$.
It can be noticed that $\mathcal{P}$ and $\mathcal{Q}$ are necessarily hermitian and positive definite (Propositon~\ref{defpos} given in the appendix).

If now, the length of the sliding window is equal to the period $T$, then, as $P$, $A$, $B$, $Q$ and $R$ are $T$-periodic, $\mathcal{P}$, $\mathcal{A}$, $\mathcal{B}$, $\mathcal{Q}$ and $\mathcal{R}$ are constant and $\dot {\mathcal{P}}=0$.
We conclude that $\mathcal{P}$ is the unique hermitian positive definite solution of the algebraic equation and we have necessarily $\mathcal{P}(t)=\mathcal{P}(0)=\mathcal{T}_T(P)(0)$.

Conversely, since the unique solution of the harmonic algebraic equation is independent of the initial condition (clearly in $H$ by Corollary~\ref{equi}), if $P(t)$ is a representative of $\mathcal{P}$, i.e. $\mathcal{P}=\mathcal{T}_T(P),$ it is unique. Moreover, as Theorem~\ref{CNS} states that this representative satisfies the periodic Riccati differential equation, $P(t)$ is absolutely continuous. Hence, $\mathcal{P}$ is a bounded operator on $\ell^2$ (Property~\ref{borne} given in the appendix).
\end{proof}

\begin{corollary}\label{lyap2}
The following are equivalent:
\begin{enumerate}
\item the function $\mathcal{V}(X)=X^*\mathcal{P}X$ is a Lyapunov function for the harmonic system in closed loop ($U=-\mathcal{K}X$ with $\mathcal{K}=\mathcal{R}^{-1}\mathcal{B}^*\mathcal{P}$ a block Toeplitz matrix): $$\dot X=(\mathcal{A}-\mathcal{B}\mathcal{K}-\mathcal{N})X,$$ 
\item the function $v(x)=x'Px$ is a Lyapunov function for the $T-$periodic linear system in closed loop ($u=-K(t)x$ with $K(t)=R^{-1}(t){B}'(t)P(t)$):
$$\dot x=(A(t)-B(t)K(t))x,$$ 
\item the functional
\begin{align*}
v(x)(t)&=<x,Px>_{L^2([t-T,t])}(t)\\&=\frac{1}{T}\int_{t-T}^tx(\tau)'P(\tau)x(\tau)d\tau,
\end{align*} is a Lyapunov function for the $T-$periodic linear system in closed loop ($u=-K(t)x$ with $K(t)=R^{-1}(t){B}'(t)P(t)$):
$$\dot x=(A(t)-B(t)K(t))x.$$ 
\end{enumerate}
\end{corollary}
\begin{proof}
First, as it is required by Proposition \ref{static} and using Property \ref{product} given in the appendix, it can be noticed that the static gain $\mathcal{K}=\mathcal{R}^{-1}\mathcal{B}^*\mathcal{P}$ is guaranteed to be always block Toeplitz. The remaining part of the proof follows similar steps as in the proof of Corollary~\ref{lyap}.
\end{proof}

%

\subsection{Stabilization of bilinear affine systems with periodic exogenous inputs} 

This subsection is dedicated to the design of stabilizing state feedback control for bilinear affine systems described by \eqref{affine} or equivalently by their harmonic model \eqref{affineh}. For these systems and for a given period $T$, the set of harmonic equilibrium points can be defined by:
\begin{align}
\mathcal{E}_T &= \{(X^e,S^e,W^e) ~|~ 0 = (\mathcal{A}(S^e)-\mathcal{N}) X^e + \mathcal{B}(S^e)W^e\}.\nonumber
\end{align}
Note that all these equilibrium points are well defined in $H$ following the same reasoning of Corollary~\ref{equi}.

\begin{proposition}\label{lawaffine} Assume that $w$ is a $T$-periodic function of class $L^{2}$. For any equilibrium point $(X^e,S^e,W) \in \mathcal{E}_T$ of system \eqref{affineh} such that the matrix $( \mathcal{A}(S^e)-\mathcal{N})$ is Hurwitz, the state feedback control law: 
\begin{align} 				
S(t)&=S^e-\Gamma \mathcal{G}^*(\mathcal{X}(t))\mathcal{P} (X(t)-X^e),\nonumber
\end{align}
stabilizes globally and asymptotically the harmonic system~\eqref{affineh} towards $X^e$ with:
\begin{itemize}
\item $\Gamma=\mathcal{T}_T(\gamma)>0$ for a given $T$-periodic symmetric positive definite and $L^{\infty}$ function $\gamma$,
\item $\mathcal{G}(\mathcal{X}(t))= \mathcal{A}_{dep}\mathcal{X}(t) + \mathcal{B}_{dep}\mathcal{W}$ with $\mathcal{X}(t)=\mathcal{T}_T(x)(t)$, $\mathcal{W}=\mathcal{T}_T(w)$ 
\item $\mathcal{P}$ the hermitian positive definite solution of the Lyapunov equation:
\begin{equation}(\mathcal{A}(S^e)-\mathcal{N}) ^* \mathcal{P}+\mathcal{P} (\mathcal{A}(S^e)-\mathcal{N}) +\mathcal{Q}=0,\label{syl2}\end{equation}
with $\mathcal{Q}=\mathcal{T}_T(Q)$ for a given $T$-periodic symmetric positive definite and $L^{\infty}$ function $Q$.
\end{itemize}
Moreover, the state feedback control:
 \begin{equation}
s(t)=s^e(t)-\gamma(t) g(x(t))'P(t)(x(t)-x^e(t),\label{emb}\end{equation} where $g(x(t))=A_{dep}x(t) + B_{dep}w(t)$, stabilizes globally and asymptotically the bilinear affine system~\eqref{affine} towards the $T$-periodic trajectory $x^e(t)=\mathcal{F}_T^{-1}(X^e)(t)$.
\end{proposition}

\begin{proof}
For a constant input $W$ and an harmonic equilibrium defined by:
\begin{equation*}
\begin{array}{rcl}
0 &=& (\mathcal{A}(S^e)-\mathcal{N})X^e+\mathcal{B}(S^e)W,
\end{array}
\end{equation*}
the harmonic model \eqref{affineh} can be rewritten as:
\begin{align}
\dot{X}(t) = (\mathcal{A}&(S^e)-\mathcal{N})(X(t)-X^e)+\cdots\nonumber\\&(\mathcal{A}(S(t))-\mathcal{A}(S^e))X(t) +\cdots\nonumber\\&(\mathcal{B}(S(t))-\mathcal{B}(S^e))W^e.\label{har}
\end{align}
It can be noticed that $$(\mathcal{A}(S(t))-\mathcal{A}(S^e))X(t)= \mathcal{A}_{dep}\mathcal{X}(t)(S(t)-S^e),$$ where $\mathcal{X}(t)=\mathcal{T}_T(x)(t)$
since 
\begin{align*}
A(s(t))x(t) &-A(s^e(t))x(t)=(A_{ind}+s(t)A_{dep})x(t)-\cdots\\
 &-(A_{ind}+s^e(t)A_{dep})x(t)\\
&=A_{dep}x(t)(s(t)-s^e(t)),
\end{align*}
 which by sliding Fourier decomposition leads to the result (see Property~\ref{product} given in the appendix).
For $(\mathcal{B}(S(t))-\mathcal{B}(S^e))W^e$, using similar steps as before, we have:
 $$(\mathcal{B}(S(t))-\mathcal{B}(S^e))W^e=\mathcal{B}_{dep}\mathcal{W}(S(t)-S^e),$$
where $\mathcal{W}=\mathcal{T}_T(w)$.
 
Therefore, equation \eqref{har} can be rewritten as:
\begin{align}
\dot{X} & = (\mathcal{A}(S^e)-\mathcal{N})(X(t)-X^e)+\mathcal{G}(\mathcal{X}(t))(S(t)-S^e)\label{har2}
\end{align}
with $\mathcal{G}(\mathcal{X}(t))= \mathcal{A}_{dep}\mathcal{X}(t) + \mathcal{B}_{dep}\mathcal{W}$.
As $(\mathcal{A}(S^e)-\mathcal{N})$ is assumed to be Hurwitz, there exists an hermitian positive definite matrix $\mathcal{P}$ solution of the harmonic Lyapunov equation:
$$(\mathcal{A}(S^e)-\mathcal{N}) ^* \mathcal{P}+\mathcal{P} (\mathcal{A}(S^e)-\mathcal{N}) +\mathcal{Q}=0,$$
for any $\mathcal{Q}=\mathcal{T}_T(Q)$ defined by a $T$-periodic symmetric positive definite and $L^{\infty}$ function $Q$. 

Consider the state feedback law:
$$S(t)=S^e-\Gamma \mathcal{G}^*(\mathcal{X}(t))\mathcal{P}(X(t)-X^e),$$ where $\Gamma=\mathcal{T}_T(\gamma)$ is defined by a $T$-periodic symmetric positive definite and $L^{\infty}$ function $\gamma$, 
it is straightforward to prove that: $$\mathcal{V}(X-X^e)=(X-X^e)^*\mathcal{P}(X-X^e),$$ is a Lyapunov function for \eqref{har2}. Indeed, the time derivative of $\mathcal{V}$ along the system trajectories satisfies: 
\begin{align*}
\dot{\mathcal{V}}&(X-X^e)=-(X-X^e)^*\mathcal{Q}(X-X^e)-\ \cdots\\
&(X-X^e)^*\mathcal{P}\mathcal{G}(\mathcal{X}(t))\Gamma\mathcal{G}^*(\mathcal{X}(t))\mathcal{P}(X-X^e)<0, \text{ if }X \neq X^e\end{align*}

Therefore, one concludes that the functional $v(x-x^e)$ is a Lyapunov function (item 3. Corollary~\ref{lyap}) for the bilinear affine system \eqref{affine}
with $$s(t)=s^e(t)-\gamma(t) g(x(t))'P(t)(x(t)-x^e(t)),$$ and where $g(x(t))=A_{dep}x(t) + B_{dep}w(t)$.
In other words, the state feedback law $s(t)$ stabilizes globally and asymptotically the bilinear affine system \eqref{affine}
towards the $T-$periodic trajectory $x^e(t)$ with $x^e=\mathcal{F}_T^{-1}(X^ e)$. 
 \end{proof}
\subsection{Stabilization of bilinear affine systems with periodic exogenous inputs and disturbances} 
Exogenous disturbances lead to deviations from the nominal operating point or to undesirable behaviours on the controlled variables. Forwarding control technics as in \cite{Astolfi2017}, \cite{Astolfi2019}, based on internal model principle \cite{Francis}, are efficient to prevent these drawbacks. In this subsection, we show how the results of these papers can be used to extend forwarding control technics to $T-$periodic systems by considering the harmonic model rather than the time domain one. To this end, we introduce an augmented state $(X,Z)$ whose dynamics is given by:
\begin{align}
\dot X(t)&=(\mathcal{A}-\mathcal{N})(X(t)-X^e)+\mathcal{G}(X(t))(S(t)-S^e)\label{e1}\\
\dot Z(t)&=(\mathcal{O}-\mathcal{N}) Z(t)+\mathcal{L}\mathcal{C}(X(t)-X^e).\label{e2}
\end{align}
Note that the first equation corresponds to equation \eqref{har} with $\mathcal{A}=\mathcal{A}(S^e)$ and it describes a bilinear affine system \eqref{affine}. 
For the second equation, we assume that there exist respectively a $T-$periodic and $L^{\infty}$ representative $L$ and $C$ for $\mathcal{L}$ and $\mathcal{C}$ as well as a $T-$periodic skew symmetric and $L^{2}$ matrix function $O(t)$ such that $\mathcal{O}=\mathcal{T}_T(O)$ is skew hermitian i.e. $-\mathcal{O}=\mathcal{O}^*$.

Now, let us show that equation \eqref{e2} allows to incorporate in the design integral actions or oscillators in order to reject undesirable harmonics. For example, if $O(t)$ takes the form:
$$O(t)=\left(\begin{array}{ccc}0 & 0 & 0 \\0 & 0 & -k\omega \\0 & k\omega & 0\end{array}\right)$$ and $(y_1,y_2,0)'=L(t)C(t)(x-x^e)$ then, 
an integral action is introduced with input $y_1$ on the first line and an oscillator with input $y_2$ on the second line. The aim of this oscillator is to reject an exogenous constant disturbance that acts on the harmonic model and disturbs the $k-$th phasor of $y_2$ (see \cite{Astolfi2019}).
\color{black}

\begin{proposition}[Forwarding control]\label{forwarding}
The state feedback control law:
\begin{align}
S(t)=S^e- \eta \mathcal{G}(X(t))^*&\left[\mathcal{P} (X(t)-X^e)-\cdots \right.\nonumber\\&\left .\mathcal{M}^*(Z(t)-\mathcal{M}(X(t)-X^e))\right] \nonumber
\end{align}
where $\eta>0$, $\mathcal{P}$ determined as in Proposition~\ref{lawaffine} and $\mathcal{M}$ satisfies 
the harmonic Sylvester equation:
\begin{equation}
(\mathcal{O}-\mathcal{N}) \mathcal{M}-\mathcal{M}(\mathcal{A}-\mathcal{N})+\mathcal{L}\mathcal{C}=0\label{syl}\end{equation}
stabilizes globally and asymptotically the equilibrium $X^e$ of \eqref{e1}-\eqref{e2}. 
\end{proposition}
\begin{proof}First, notice that an equivalence between the solution of the harmonic Sylvester differential equation and its time $T-$periodic version can be proved using similar steps as in the proof of Theorem~\ref{Lyapunov}.
Let $\tilde X=(X-X^e)$, and consider the Lyapunov function candidate defined by:
$$\mathcal{V}(\tilde X,Z)=\tilde X^*\mathcal{P} \tilde X+(Z-\mathcal{M} \tilde X)^*(Z-\mathcal{M} \tilde X).$$
Then, the derivative of $\mathcal{V}$ along the trajectories leads to:
\begin{align*}\dot{\mathcal{V}}(\tilde X,Z)=-\tilde X^*&\mathcal{Q} \tilde X+2\tilde X \mathcal{P}\mathcal{G}(X)\tilde S+\cdots \\&2(Z-\mathcal{M} \tilde X)^*((\mathcal{O}-\mathcal{N}) Z+\cdots \\&\mathcal{L}\mathcal{C}\tilde X-\mathcal{M} ((\mathcal{A}-\mathcal{N})\tilde X+\mathcal{G}(X)\tilde S)).\end{align*}
where $\tilde S=(S-S^e)$.
Using the relation $(\mathcal{O}-\mathcal{N}) \mathcal{M}-\mathcal{M}(\mathcal{A}-\mathcal{N})+\mathcal{L}\mathcal{C}=0$, we have:
\begin{align*}\dot{\mathcal{V}}(\tilde X)=-\tilde X^*&\mathcal{Q} \tilde X+\cdots\\
&2\left (\tilde X ^*\mathcal{P}\mathcal{G}(X)-(Z-\mathcal{M} \tilde X)^*\mathcal{M}\mathcal{G}(X)\right)\tilde S+\cdots\\
&2(Z-\mathcal{M}\tilde X)^*(\mathcal{O}-\mathcal{N})(Z-\mathcal{M}\tilde X).
\end{align*}
As $(\mathcal{O}-\mathcal{N})^*=-(\mathcal{O}-\mathcal{N})$, one gets for any $Y$, $$Y^*(\mathcal{O}-\mathcal{N})Y=0,$$ 
and the derivative reduces to:
\begin{align*}\dot{\mathcal{V}}(\tilde X,Z)&=-\tilde X^*\mathcal{Q} \tilde X+2\left (\tilde X ^*\left[\mathcal{P}-(Z-\mathcal{M} \tilde X)^*\mathcal{M}\right]\mathcal{G}(X)\right)\tilde S.
\end{align*}
Taking $$\tilde S= - \eta \mathcal{G}(X)^*\left[\mathcal{P}\tilde X-\mathcal{M}^* (Z-\mathcal{M}\tilde X)\right]$$ for any $\eta>0$ yields the result:
$$\dot{\mathcal{V}}(\tilde X,Z)=-\tilde X^*\mathcal{Q} \tilde X-\frac{1}{\eta}\tilde S^* \tilde S<0, \ \tilde X\neq 0.$$
Moreover, as $\dot{\mathcal{V}}=0$ implies $\tilde X= 0$, it follows that $\mathcal{V}(0,Z)=Z^*Z$ is constant and we conclude that $Z$ is bounded and that $\mathcal{V}(\tilde X,Z)$ is a weak Lyapunov function.

\end{proof}
The following corollary follows easily.
\begin{corollary}\label{cont2}
The state feedback control law:
\begin{align}
s(t)=s^e(t)-&\eta g(x(t))'\left[P(t)\tilde x(t)-\cdots\right. \nonumber\\
&\left.-M(t)'(z(t)-M(t)\tilde x(t))\right] \label{fb}
\end{align}
with $\tilde x(t)=x(t)-x^e(t)$, $g(x(t))=A_{dep}x(t) + B_{dep}w(t)$ and where $P(t)$ and $M(t)$ are respectively the representative of $\mathcal{P}=\mathcal{T}_T(P)$ and $\mathcal{M}=\mathcal{T}_T(M)$ (see Proposition~\ref{forwarding})
stabilizes globally and asymptotically the system: 
\begin{align}
\dot { x}(t)&=A(s(t)) x(t)+B( s(t)) w(t)\nonumber\\ 
\dot z(t)&=O(t)z(t)+L(t)C(t)\tilde x(t) \nonumber
\end{align}
towards the $T$-periodic trajectory $x^e(t)=\mathcal{F}_T^{-1}(X^e)(t)$.
\end{corollary}
\begin{proof}
The proof is established invoking Theorem~\ref{CNS}.
\end{proof}

The following proposition is of particular interest when saturations occur. 
\begin{proposition}\label{slaw}
If the control domain is restricted to a closed set $S_c$ and if the $T-$periodic trajectory $s^e$ belongs to the interior of $S_c$, the saturation function
\begin{equation}sat(s(t))=s^e(t)+\alpha(t) (s(t)-s^e(t))\label{sat} \end{equation}
where $\alpha(t) =max\{0\leq\alpha \leq1: s^e(t)+\alpha \tilde s(t)\in S_c\}$, does not modify the stability property of the control law of Corollary~\ref{cont2}.
\end{proposition}
\begin{proof}
If the $T-$periodic trajectory $s^e$ belongs to the interior of $S_c$, there always exists, at any time $t$, an unique strictly positive $\alpha(t)\leq 1$ 
such that $\alpha(t) =\max\{0\leq\alpha \leq1: s^e(t)+\alpha (s(t)-s^e(t))\in S_c\}$. 
Let us consider the positive definite functional (as in Corollary \ref{lyap}):
\begin{align*}
v(\tilde x,z)(t)&=<\tilde x,P\tilde x>_{L^2([t-T,t])}(t)\ \cdots\\
&+<z-M\tilde x,z-M\tilde x>_{L^2([t-T,t])}(t)
\end{align*} 
where $\tilde x=x-x^e$.
When the feedback law \eqref{fb} is saturated as in \eqref{sat}, the derivative along the trajectories is given by:
\begin{align*}
&\dot v(x,z)(t)=-<\tilde x,Q\tilde x>_{L^2([t-T,t])}(t)\ \cdots\\
&-\eta<g(x)'P\tilde x,\alpha g(x)'P\tilde x>_{L^2([t-T,t])}(t)\ \cdots\\
&-\eta <g(x)'M'(z-M\tilde x),\alpha g(x)'M'(z-M\tilde x)>_{L^2([t-T,t])}(t)\end{align*}
As $\alpha(t)>0$, the result follows.
\end{proof}

\section{Application to a single-phase rectifier bridge}
In this part, we consider a single-phase rectifier bridge (see Fig.\ref{sch_red1P_appli}) whose primary objective is to supply a near constant voltage $v_{dc}$ to a given resistive load from an AC voltage source $v_{in}$. In the sequel, $i$ refers to the current and $s$ to the control variable. 
Considering a resistive load, $i_{dc}=\frac{v_{dc}}{R_L}$, leads to the following states equations:\begin{align*}
	\dot {\begin{bmatrix} i(t)\\ v_{dc}(t)\end{bmatrix}} &=
	\begin{bmatrix}
		\frac{-R}{L} 		& \frac{-1}{L}s(t)\\
		\frac{1}{C}s(t)		 	& \frac{-1}{R_{L}C}
	\end{bmatrix}
	\begin{bmatrix} i(t) \\ v_{dc}(t)		\end{bmatrix} + 
	\begin{bmatrix}
		\frac{1}{L} \\
		0 
	\end{bmatrix}
	v_{in}(t) 
\end{align*}
It is assumed that the nominal input voltage is known and given by: $v_{in}=100 sin(\omega t)$ with $\omega =2\pi f$ and $f=50hz$. The parameters for the converter are the following: $R = 1 m \Omega$, $L = 1 mH$, $C = 4 mF$ and $R_L = 10 \Omega$. It is obvious that this system is a switched affine system for which the control design problem can be formulated as a control problem with continuous variables invoking density arguments; see \cite{Beneux} for example. Thus, we assume in the sequel that the control $s$ takes its values in the set $[-1 \ 1]$ instead of the discrete set $\{-1\ ,1\}$. Moreover, the harmonic model can be established as follows (modulo a permutation matrix):
\begin{align*}
	\dot {\begin{bmatrix} I(t)\\ V_{dc}(t)\end{bmatrix}} &=
	\begin{bmatrix}
		\frac{-R}{L}\mathcal{I}-\mathcal{N} 		& \frac{-1}{L}\mathcal{S}(t)\\
		\frac{1}{C}\mathcal{S}(t)		 	& \frac{-1}{R_{L}C}\mathcal{I}-\mathcal{N}
	\end{bmatrix}
	\begin{bmatrix} I(t) \\ V_{dc}(t)		\end{bmatrix} + 
	\begin{bmatrix}
		\frac{1}{L}\mathcal{I} \\
		\mathcal{Z} 
	\end{bmatrix}
	V_{in}(t) 
\end{align*}
where $\mathcal{S}=\mathcal{T}_T(s)$, $I=\mathcal{F}_T(i)$, $V_{in}=\mathcal{F}_T(v_{in})$, $V_{dc}=\mathcal{F}_T(v_{dc})$, $\mathcal{I}=\mathcal{F}_T(I)$ is the infinite dimensional identity matrix and $\mathcal{Z}$ is the infinite dimensional zero matrix.

\begin{figure}
\includegraphics[scale=0.5]{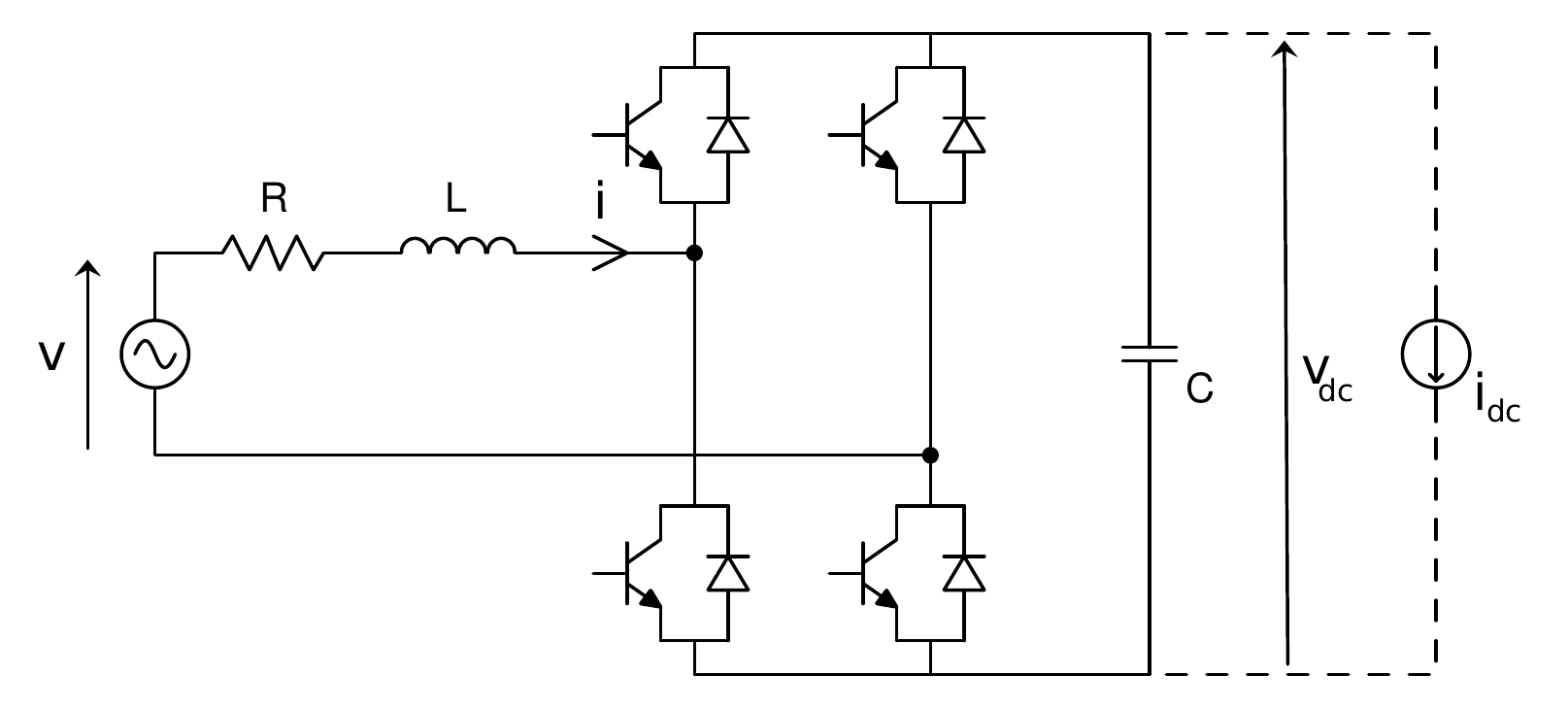}
\caption{Electrical scheme of single-phase rectifier bridge.}
\label{sch_red1P_appli}
\end{figure}

The control objective is twofold : reduce the signals harmonic content and ensure a $T-$periodic steady state with
\begin{itemize}
\item the mean value of $v_{dc}$ around $v_{ref}=200$ volts with values within the band $\pm 10 \%$,
\item the power factor maximized which occurs when the phase between signal $v_{in}$ and $i$ is minimized.
\end{itemize}
Regarding harmonic content reduction, it is an important objective from the practical point of view. It allows to avoid harmonic pollution of the electrical network from which the input voltage is derived.

\subsection{$T-$periodic equilibrium determination} 

The first task is to determine an equilibrium $(V^e,I^e,S^e)$ for the harmonic model that feets the control objectives. This is done by defining a criterion with a weighting sum of the the following quantities: 
\begin{itemize}
\item $|V_{dc_0}^e-v_{ref}|^2\approx 0$ to satisfy the objective on the DC component of the voltage.
\item $\sum_{|k|\geq 1} |V_{dc_k}^e|^2$ which has to be sufficiently small in order to reduce the harmonic content of the voltage.
\item $\sum_{|k|\geq 1} |I^e_k|^2$ which has to be sufficiently small in order to reduce harmonic content of the current.
\item $\sum_{|k|\geq 0}|Real(I_k^e)|^2\approx 0$ to force $v_{in}$ and $i$ to be in phase ($v_{in}$ is an odd function which means that $Real(v_{in})=0$).
\end{itemize}
Notice that as the signals have to be real, the relation $X_{-k}= \bar X_{k}$ must be satisfied and a substitution can be made for harmonics of negative subscript to simplify the problem.
As a consequence, the determination of the equilibrium is carried out using the following objective function :
\begin{align*}J^*&=\min_{S} \left( w_0 |V_{dc_0}^e-V_{ref}|^2+w_1\sum_{k\geq 1} |V_{dc_k}^e|^2+\cdots \right. \\
&\left. w_2\sum_{|k|\geq 0}|Re(I_k^e)|^2+w_3\sum_{k\geq 2} |I^e_k|^2 \right) \end{align*}
where $I^e$ and $V^e$ are determined for a given $\mathcal{S}$ by the equilibrium equation:
\begin{align*}
\begin{bmatrix} I^e \\ V_{dc}^e	\end{bmatrix}=-	\begin{bmatrix}
 \frac{-R}{L}\mathcal{I}-\mathcal{N} 		& \frac{-1}{L}\mathcal{S}\\
 \frac{1}{C}\mathcal{S}		 	& \frac{-1}{R_{L}C}\mathcal{I}-\mathcal{N}
\end{bmatrix}^{-1} 
\begin{bmatrix}
 \frac{1}{L}\mathcal{I} \\
 \mathcal{Z} 
\end{bmatrix}
 V_{in} 
\end{align*}

This nonlinear least squares problem can be easily solved for any weighting matrix $W=[w_0,w_1,w_2,w_3]$ using ad hoc solver. 
\subsection{State feedback control} 
Once the equilibrium is determined, the control law given in Proposition~\ref{lawaffine} can be used to stabilize asymptotically the closed loop system. In practice, a truncation is necessary on the maximum order $h_{max}$ of the harmonics to be taken into account. A detailed analysis of the approximation introduced by this truncation is out of the scope of this paper. Further developments are necessary to address this very challenging problem. For simulation purposes, a truncation corresponding to $h_{max}=5$ is considered as being satisfactory for this example. Indeed, higher orders do not modify or improve the obtained simulation results. 

We choose a weighing matrix $W=[w_0,w_1,w_2,w_3]=[10^3,1,10^3,1]$. The matrix $Q(t)$ involved in the Lyapunov equation \eqref{syl2} is taken constant : $Q(t)=diag([Q_i, Q_v])$ with $Q_i=Id_n$ for the current part and $Q_v=\frac{1}{10^2}Id_n$ for the voltage part. Finally, $\gamma(t)$ (see \eqref{emb}) is taken constant and equal to $10^{-4}$. 
 
 Figure~\ref{f11} shows respectively the current $i$, the output voltage $v_{dc}$ and the control $s$ for a startup of the converter from zero initial conditions. As it can be observed and as expected, all these variables converge to the equilibrium $i^e$, $v^e$ and $s^e$ determined by solving the previous nonlinear optimization problem. The power factor is maximized and the output voltage meets the objective in mean value $v_{dc}=200$ volts $\pm 10\%$. Moreover, the harmonic content of both current and voltage is reduced. It can be observed that, after starting the converter, the control is saturated to $1$. As it is proved in Corollary \ref{slaw}, this saturation does not affect the stability property.
\begin{figure}[h]
\includegraphics[width=\linewidth,height=7cm]{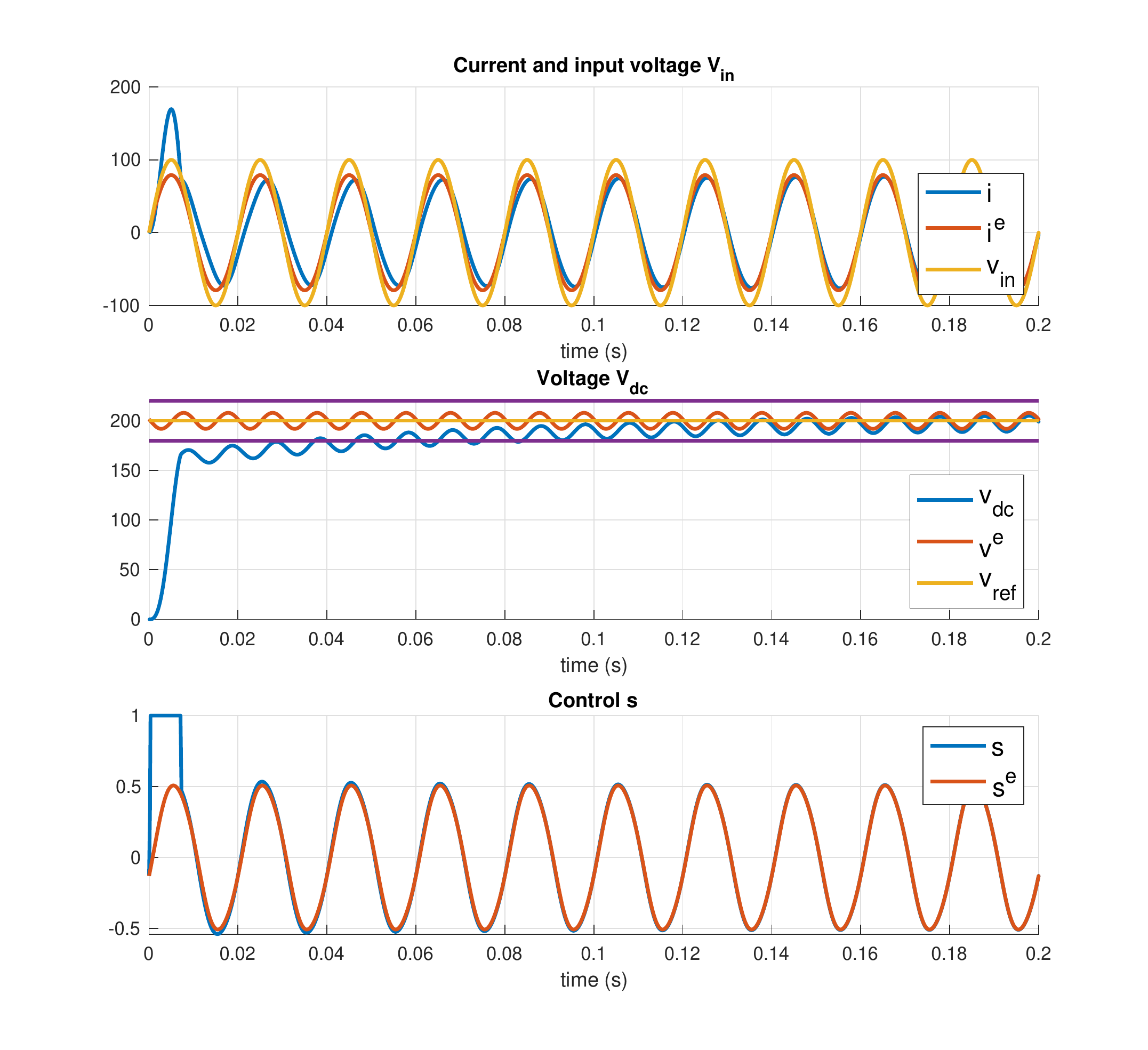}
\caption{Startup of the converter from zero initial conditions}
\label{f11}
\end{figure}

\subsection{Perturbed system and Integral actions}
In this section, two types of $T-$periodic disturbances are considered: an input disturbance $V_{{in}_p}$ on $V_{in}$ and a load current disturbance $I_p$. The harmonic model is now described by the following equations:
\begin{align*}
\dot {\begin{bmatrix} I(t)\\ V_{dc}(t)\end{bmatrix}} &=
\begin{bmatrix}
 \frac{-R}{L}\mathcal{I}-\mathcal{N} 		& \frac{-1}{L}\mathcal{S}(t)\\
 \frac{1}{C}\mathcal{S}(t)		 	& \frac{-1}{R_{L}C}\mathcal{I}-\mathcal{N}
\end{bmatrix}
\begin{bmatrix} I(t) \\ V_{dc}(t)		\end{bmatrix} +\cdots \\ 
&\begin{bmatrix}
 \frac{1}{L}\mathcal{I} & \mathcal{Z} \\
 \mathcal{Z} & \frac{1}{C}\mathcal{I}
\end{bmatrix}
\left[\begin{array}{c}V_{in}+ V_{{in}_p} \\I_p\end{array}\right].
\end{align*}

In order to preserve the control objectives against these disturbances, integral actions can be considered in the control design provided that augmented system remains stabilizable. Hence, the choices and the number of these integral actions are constrained. To preserve the mean value of the output voltage, we add an integral action on the dc component of output voltage $V_{{dc}_0}$. To force the current to be in phase with the input voltage in order to maximize the power factor, a second integral action is introduced (recall that $V_{in}$ is assumed to be an odd function, the integral action forces the real part of the first harmonic $(I_1-I^e_1)$ to zero). As a result, these integral actions are taken into account by considering the new variable $Z=(Z_1,Z_2)$ whose dynamics is given by:
\begin{align}
\dot Z&=-\mathcal{N} Z+\left(\begin{array}{c}\mathcal{L}_1 \mathcal{C}_1\\\mathcal{L}_2\mathcal{C}_2\end{array}\right)(X-X^e) \label{int}
\end{align}
with $\mathcal{C}_1=[\mathcal{Z}\ \mathcal{I}]$, $\mathcal{C}_2=[\mathcal{I}\ \mathcal{Z}]$,
$\mathcal{L}_1=\mathcal{T}_T(\gamma_1)$ and : 
\begin{align*}\mathcal{L}_2&=\gamma_2\left[
\begin{array}{ccccc}
\ddots & & \vdots & &\udots \\ 
&0 & 1 & 0 & \\
\cdots & 1& 0 & 1 & \cdots \\
 & 0 & 1 & 0& \\
\udots & & \vdots & & \ddots\end{array}\right]
\\&=\mathcal{T}_T(\gamma_2 cos(\omega t))
\end{align*}
and where $\gamma_1 >0$ and $\gamma_2>0$ are constant gains to be tuned.
The above equation corresponds to equation \eqref{e2} with $\mathcal{O}$ identically equals to zero.

Applying the forwarding control design proposed in Corollary \ref{cont2}, a controller is obtained with $\gamma_1=200$ and $\gamma_2=100$.
The state feedback is then given by:
\begin{align}s(t)=&sat(s^e(t)-\eta_1 g(x)'P(t)\tilde x-\cdots\nonumber\\
&\eta_2 g(x)'M(t)'(z-M(t)\tilde x)),\end{align} 
where $\tilde x(t)=x(t)-x^e(t)$ and where the saturation function is defined by:
$$sat(s)=\begin{cases}
1 \text { if }s>1\\
-1\text{ if }s<-1\\
s \text{ otherwise}
\end{cases}.$$ 
The matrix $P(t)$ is the same as in the previous section and $M(t)$ is provided by the Sylvester equation~\eqref{syl}.
The matrix $g(x)$ is given by: $g(x)=\left(\begin{array}{cc}0 & -\frac{1}{L} \\\frac{1}{C} & 0\end{array}\right)x$. Finally, the tuning parameters $\eta_1$ and $\eta_2$ are chosen to be equal to $\eta_1=10^{-7}$, $\eta_2=2.10^{-9}$.

For simulation purpose, we introduce the following disturbances at time $t=0.04$:
\begin{itemize}
\item the input voltage $V_{in}(t)=100 sin(\omega t)$ is perturbed by $V_{{in}_p}(t)=10 sin(\omega t)+20 sin(3\omega t)+20 sin(5\omega t)$,
\item the load current is perturbed by adding 2th and 4th order harmonics: $i_{per}(t)= 20 cos (2\omega t)- 20 sin(4\omega t)$. This introduces harmonic current disturbances on the odd harmonics.
\end{itemize}

\begin{figure}[h]
\includegraphics[width=\linewidth,height=7cm]{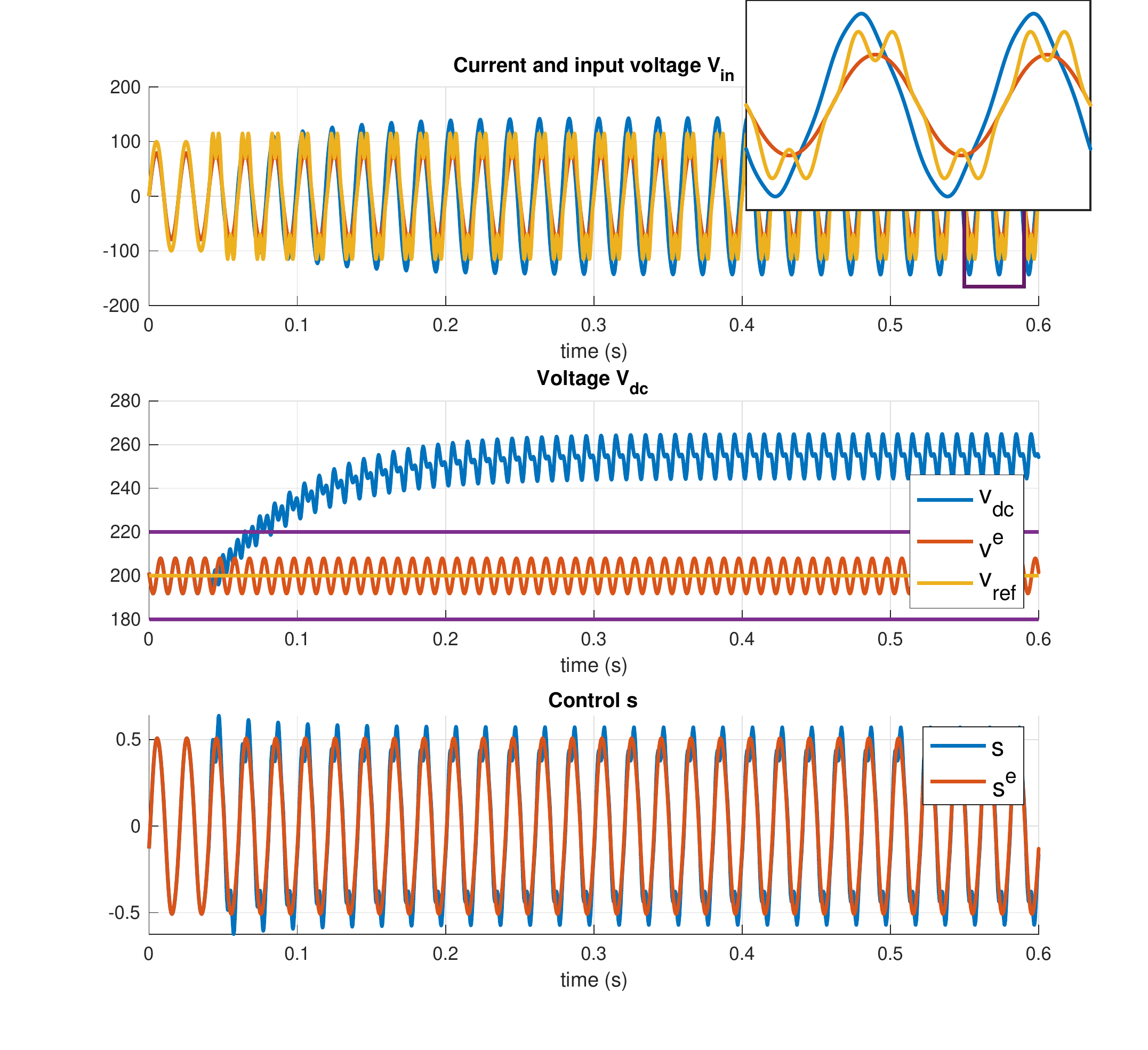}
\caption{Disturbance response when no integral action occurs.}
\label{f12} 
\end{figure}
\begin{figure}[h]
\includegraphics[width=\linewidth,height=7cm]{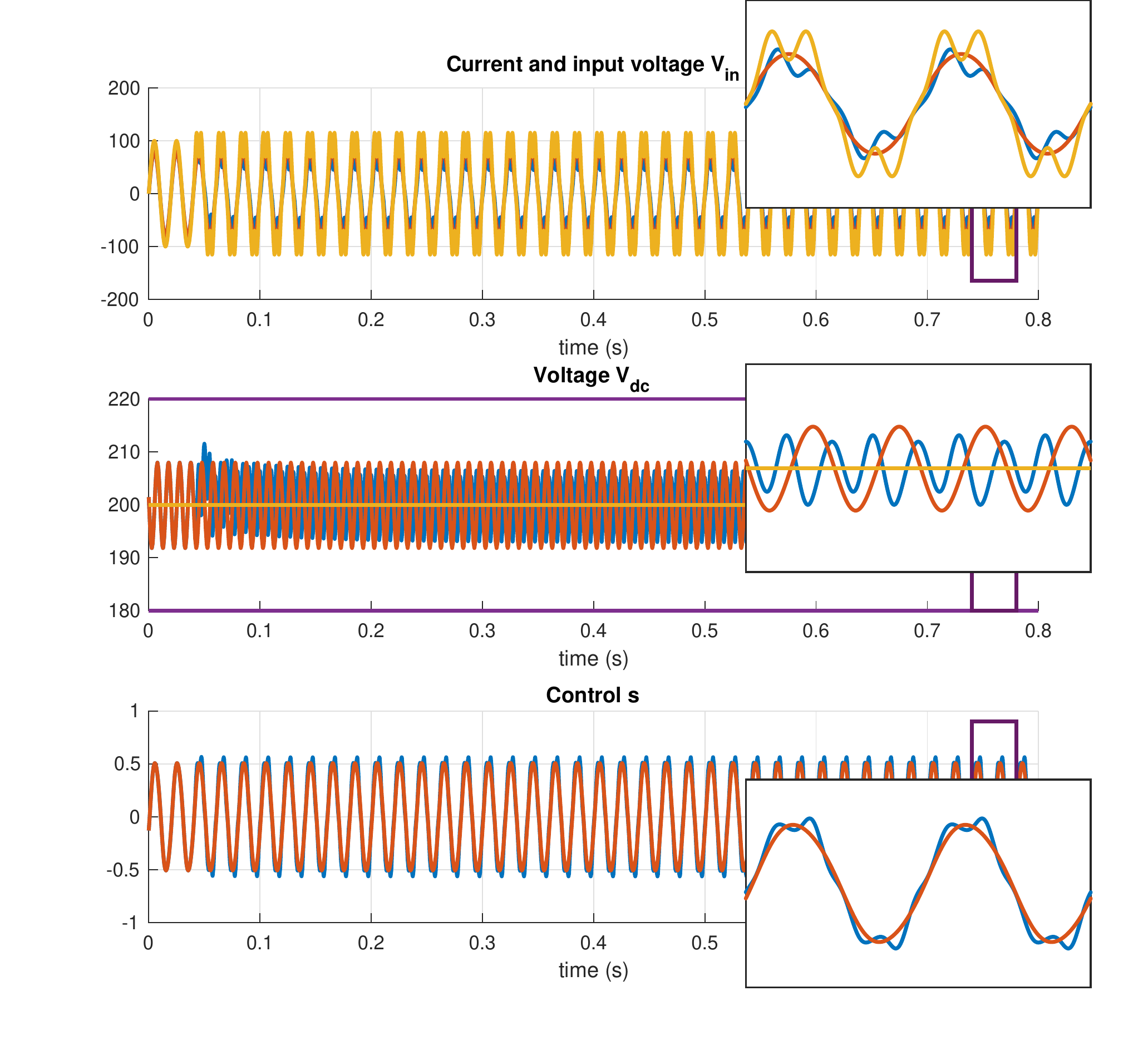}
\caption{Disturbance rejection with integral actions.}
\label{f13} 
\end{figure}
Figure~\ref{f12} shows the effect of these disturbances when no integral action is introduced. As it can be seen, the voltage is clearly out of its nominal reference while the current is also out of phase with respect to $V_{in}$.

Now, if we consider the two integral actions, the control objectives are satisfied as shown in Figure~\ref{f13} where we also see that the current harmonic content (blue curve) is not satisfactory.

We decided to add two additional integral actions on the phasors $I_3$ and $I_5$ of $I$ by defining the new variable $Z=(Z_1,Z_2,Z_3,Z_4)$ where $Z_1$ and $Z_2$ are defined by \eqref{int} and where the additional variables $(Z_3,Z_4)$ are defined by:
\begin{align}
\dot Z_i&=(\mathcal{O}_i -\mathcal{N})Z_i+\mathcal{L}_i \mathcal{C}_i(X-X^e) \label{int2}
\end{align}
for $i=3,4$, $\mathcal{O}_i=\left[
\begin{array}{cc}
\mathcal{Z} & \mathcal{T}_T(-j\omega h_i)\\
\mathcal{T}_T(j\omega h_i)& \mathcal{Z}
\end{array}\right]
$, with $h_3=3$ and $h_4=5$, $\mathcal{C}_i=\mathcal{C}_2$ and $\mathcal{L}_i=\mathcal{T}_T(\gamma_i)$, $\gamma_i>0$.

\begin{figure}[ht]
\includegraphics[width=\linewidth,height=7cm]{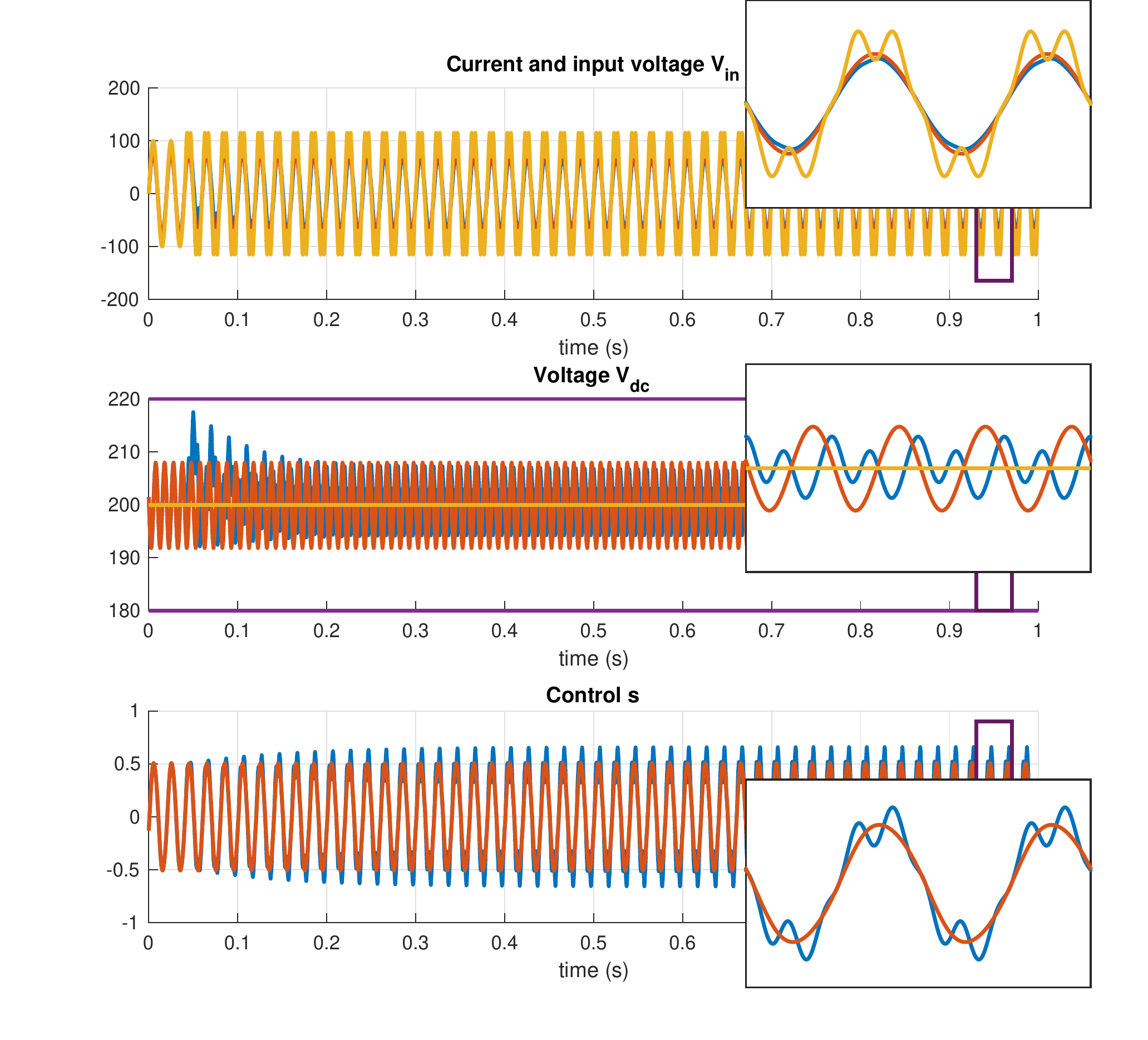}
\caption{Perturbation rejection with integral actions for primary objectives and clean current harmonic content.}
\label{f15} 
\end{figure}

\begin{figure}[ht]
\includegraphics[width=\linewidth,height=12cm]{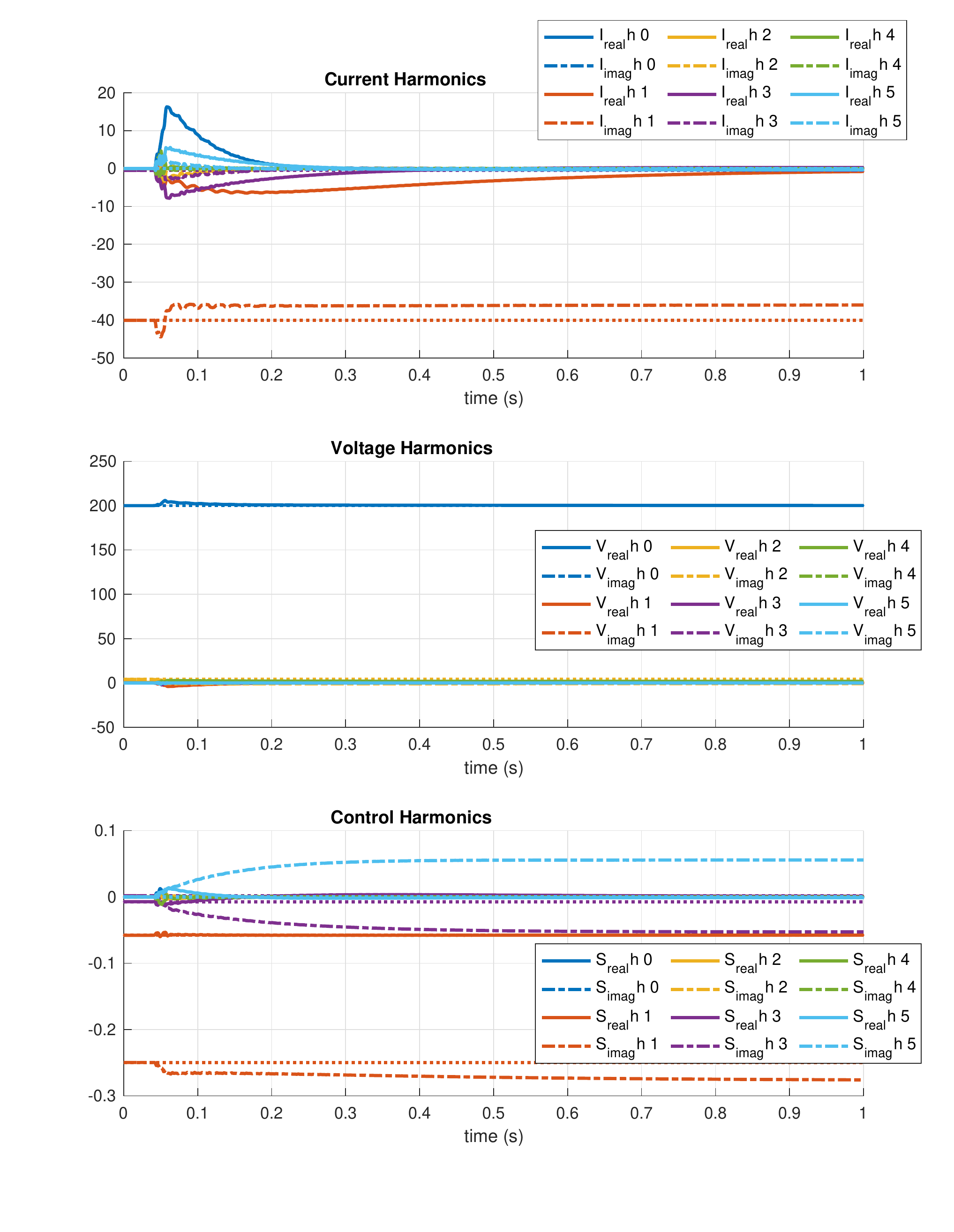}
\caption{Current, voltage and control harmonic responses for both primary objectives and clean current harmonic content. Solid line: Real part; Dash-Dot line: Imaginary part, Dotted line references.}
\label{f16} 
\end{figure}

The controller is now designed with $\gamma_1=400$, $\gamma_2=100$ $\gamma_3=1000$, $\gamma_4=2000$, $\eta_1=4.10^{-8}$ and $\eta_2 =4.10^{-10}$.
Figure~\ref{f15} shows that the control objectives are now satisfied. The current harmonic content is free from undesirable harmonics as it is confirmed by Figure~\ref{f16}. Finally, the control harmonic content has non vanishing components of order $h=1$, $h=3$ and $h=5$.

\section{Conclusion}

In this paper, a one-to-one correspondence has been established between signals $x\in L_{loc}^2$ and their sliding Fourier decomposition $X\in H$ through Theorem~\ref{coincidence} and Proposition\ref{bij}. Appropriate formulas have been provided to reconstruct a signal from its phasors (Proposition \ref{rec}). The consequence of these first important results is the proof of a strict equivalence between the Carathéodory's solutions of general nonlinear dynamical systems and their harmonic versions (Theorem~\ref{CNS}). This equivalence is an another very important result with a major impact. We illustrated this impact by deriving new interpretations of Lyapunov and Riccati harmonic equations for fairly general LTP systems without invoking the Floquet's Theorem. 
We also demonstrated the benefit of these results from a practical point of view as the design part of this paper covers several classes of systems including bilinear affine systems with application to power converters. With the contributions of this paper, one is able to design globally and asymptotically stable feedback control laws in the harmonic domain, including integral actions on the phasors and obtain the time domain version of these control laws which are of great help in practice. The results have been illustrated on a realistic model of a single-phase rectifier bridge subject to periodic exogenous inputs. This example shows clearly that the equilibrium and the control design can be seriously simplified in the harmonic domain using the theory presented in this paper to provide stabilizing control strategies with harmonic disturbances rejection.
This paper leaves open some challenging questions such as the order of the truncation and the robustness with respect to parameter variations and to the period $T$. This will be addressed in future works.

\section{Appendix}

This appendix recalls properties used in this paper and borrowed from Toeplitz (Laurent) decomposition Part V p.p. 562-573 of \cite{Gohberg}. By convention, the components of $X$ are stored in decreasing order: $X=\left(\cdots, X_{-1}', X_0', X_{1}',\cdots \right)'.$
\begin{definition}\label{toeplitz}
 For any matrix function $A\in L^{1}_{loc}(\mathbb{R},\mathbb{C}^{n\times m} )$, we denote: $\mathcal{A}(t)=\mathcal{T}_T(A)(t)$ the Toeplitz matrix of infinite dimension defined by 
$$\mathcal{A}(t)=\mathcal{T}_T(A)(t)=
\left[
\begin{array}{ccccc}
\ddots & & \vdots & &\udots \\ & A_0(t) & A_{-1}(t) & A_{-2}(t) & \\
\cdots & A_{1}(t) & A_0(t) & A_{-1}(t) & \cdots \\
 & A_{2}(t) & A_{1} (t)& A_0(t) & \\
\udots & & \vdots & & \ddots\end{array}\right],$$
where $$A_k(t)=\frac{1}{T}\int_{t-T}^t A(\tau)e^{-j\omega k \tau}d\tau.$$

If $A\in L^{2}_{loc}(\mathbb{R},\mathbb{C}^{n\times m} )$ then $\mathcal{T}_T(A)(t)$ is an operator of $\ell^2$ in $\ell^2$, for any fixed $t$ (not necessarily bounded). 
\end{definition}

\begin{property}\label{borne} For any fixed $t$,
$\mathcal{T}_T(A)(t)$ is bounded on $\ell^2$ (i.e. $\sup_X\frac{\| \mathcal{T}_T(A)(t)X \|_{\ell^2}}{\| X \|_{\ell^2}}<+\infty$) if and only if $A\in L^{\infty}_{loc}(\mathbb{R},\mathbb{C}^{n\times m} )$.
\end{property}
The following proposition is useful for the computation of sliding Fourier decomposition of any polynomial system.
\begin{property}\label{product}
For any scalar function $\lambda$ and any $n$-dimensional vector function $x$ such as $\lambda$, $x$ and $\lambda x$ admit a Fourier decomposition, we have: $$\mathcal{F}_T(\lambda x )=(\mathcal{T}_T(\lambda)\otimes Id_n)X,$$ where $\otimes$ is the kroneker product, $Id_n$ is the $n$-dimensional identity matrix and $X=\mathcal{F}_T(x)$.

For any matrix function $A$ and vector function $x$ such as $x$, $A$ and $Ax$ admit a Fourier decomposition, we have: $$\mathcal{F}_T(Ax)=\mathcal{A}X,$$ with $\mathcal{A}=\mathcal{T}_T(A )$ and $X=\mathcal{F}_T(x)$. 

For all $A$ and $B$ matrix functions of appropriate dimensions such as $A$, $B$ and $AB$ admit a Sliding Fourier decomposition, we have : $$\mathcal{T}_T(AB)=\mathcal{T}_T(A)\mathcal{T}_T(B)=\mathcal{A}\mathcal{B}.$$
\end{property}
Note that the last property implies that the product of two Toeplitz matrices is Toeplitz, which is not the case in finite dimension.
\begin{property}[Inverse of square block Toeplitz matrix]\label{inv}
For any matrix function $A$ in $L^{\infty}_{loc}$, $\mathcal{T}_T(A)$ is invertible if there is $\gamma>0$ such that the set $\{t: |det(A(t))|<\gamma\}$ has measure zero.
The inverse $\mathcal{T}_T(A)$ is determined by $\mathcal{T}_T(A^{-1})$.
\end{property}

\begin{proposition}\label{defpos} Let $A\in L^{2}_{loc}$ be a matrix function with real values. 
\begin{enumerate}
\item $A$ is $T$-periodic if and only if $\mathcal{A}=\mathcal{T}_T(A)$ is constant
\item $A(t)=A(t)' \ a.e. $ if and only if $\mathcal{A}$ is hermitian.
\item Assume that $A\in L^{\infty}_{loc}$. Then, $A(t)> 0 \ a.e.$ if and only if $\mathcal{A}$ is positive definite.
\end{enumerate}
\end{proposition}
\begin{proof}
Point 1 is obvious.
For point 2: as the following relationship is satisfied when $A$ is a real-valued matrix function:
$\mathcal{T}_T(A')(t)=\mathcal{A}^*(t)$, we can conclude that $A(t)=A(t)'\ a.e.$ implies that the matrix $\mathcal{A}=\mathcal{T}_T(A)$ is hermitian and reciprocally.
For point 3: by generalized Hölder's inequality if $A\in L^{\infty}_{loc}$ then $Ax\in L^{2}([t-T \ t])$ for any $x\in L^2([t-T \ t])$. Thus, Parseval's Theorem~\ref{thparseval} implies:
\begin{align}
<x,Ax>_{L^2([t-T \ t])}(t)&=\frac{1}{T}\int_{t-T}^tx(\tau)'A(\tau)x(\tau)d\tau\nonumber\\
&=X^*(t)\mathcal{A}(t)X(t),\label{parseval2}\end{align}
with $X=\mathcal{F}_T(x)$. Note that the assumption $A\in L^{\infty}_{loc}$ implies that $\mathcal{A}(t)$ is a bounded operator on $\ell^2$ by property~\ref{borne}.
For any $t$ and any $x\in L^2([t-T\ t])$, if $A(t)>0 \ a.e.$, the integral in \eqref{parseval2} is zero if and only if $x$ is equal to zero almost everywhere which implies that $X^*(t)\mathcal{A}(t)X(t)=0$ if and only if $X(t)=0$ and thus that for any $t$, $\mathcal{A}(t)>0$. Conversely, if $\mathcal{A}$ is positive definite then, for any $t$, the integral in \eqref{parseval2} is strictly positive for any $x\in L^2([t-T\ t])$ such that $x\neq 0\ a.e.$ on $[t-T\ t]$ which implies necessarily that $A(t)>0 \ a.e.$.
\end{proof}

\begin{IEEEbiography}[{\includegraphics[width=1in,height=1.25in,clip,keepaspectratio]{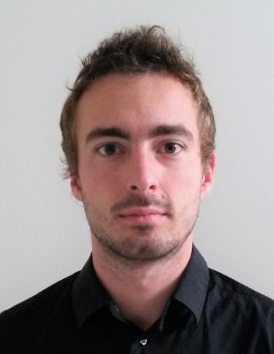}}]
{Nicolas Blin} received both an engineering diploma (Master's degree) in energy from ENSEM and a Master's degree in Systems, Information Technology and Communication from Lorraine University, France, in 2017.
Since 2018, he has been a Ph.D. student at CRAN-CNRS, Université de Lorraine and with Safran Electronics $\&$ Defense, France. His research interests are transmission and conversion of energy, power systems modelling and control, harmonics and electrical drive.
\end{IEEEbiography}
\begin{IEEEbiography}[{\includegraphics[width=1in,height=1.25in,clip,keepaspectratio]{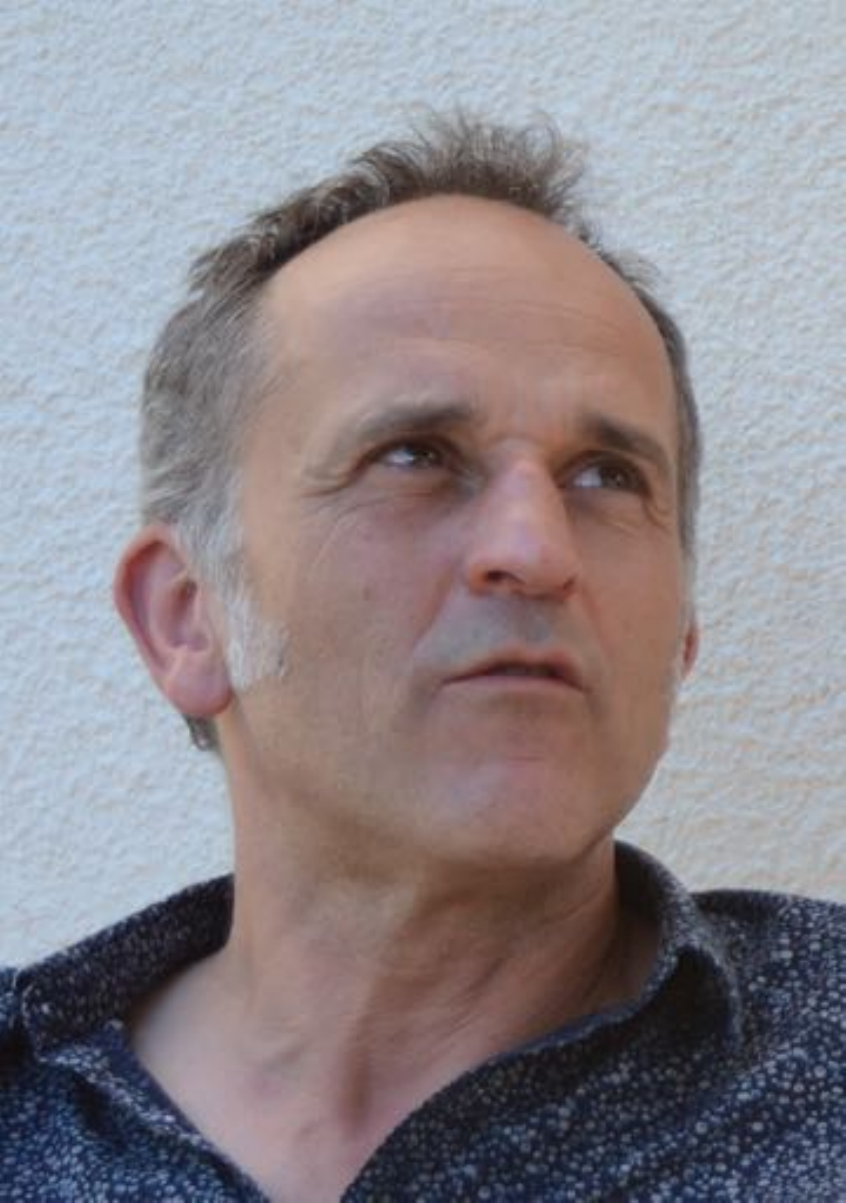}}]
{Pierre Riedinger} is a Full Professor and Head of Information Sciences Department at the engineering school Ensem and researcher at CRAN - CNRS UMR 7039, Universit\'e de Lorraine (France).
He received his M.Sc. degree in Applied Mathematics from the University Joseph Fourier, Grenoble in 1993 and the Ph.D. degree in Automatic Control in 1999 from the Institut National Polytechnique de Lorraine (INPL). He got the French Habilitation degree from the INPL in 2010. 
His current research interests include control theory and optimization of systems with their applications in electrical and power systems.
\end{IEEEbiography}
\begin{IEEEbiography}[{\includegraphics[width=1in,height=1.25in,clip,keepaspectratio]{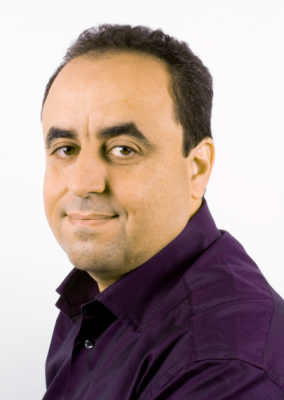}}]
{Jamal Daafouz} 
is a Full Professor at University
de Lorraine (France) and researcher at CRAN-CNRS. In 1994, he received a Ph.D.
in Automatic Control from INSA Toulouse, in 1997.
He also received the "Habilitation à Diriger des
Recherches" from INPL (University de Lorraine),
Nancy, in 2005.
His research interests include analysis, observation
and control of uncertain systems, switched
systems, hybrid systems, delay and networked systems with a particular
interest for convex based optimisation methods.
In 2010, Jamal Daafouz was appointed as a junior member of the
Institut Universitaire de France (IUF). He served as an associate editor
of the following journals: Automatica, IEEE Transactions on Automatic
Control, European Journal of Control and Non linear Analysis and Hybrid
Systems. He is senior editor of the journal IEEE Control Systems Letters. \end{IEEEbiography}
\begin{IEEEbiography}[{\includegraphics[width=1in,height=1.25in,clip,keepaspectratio]{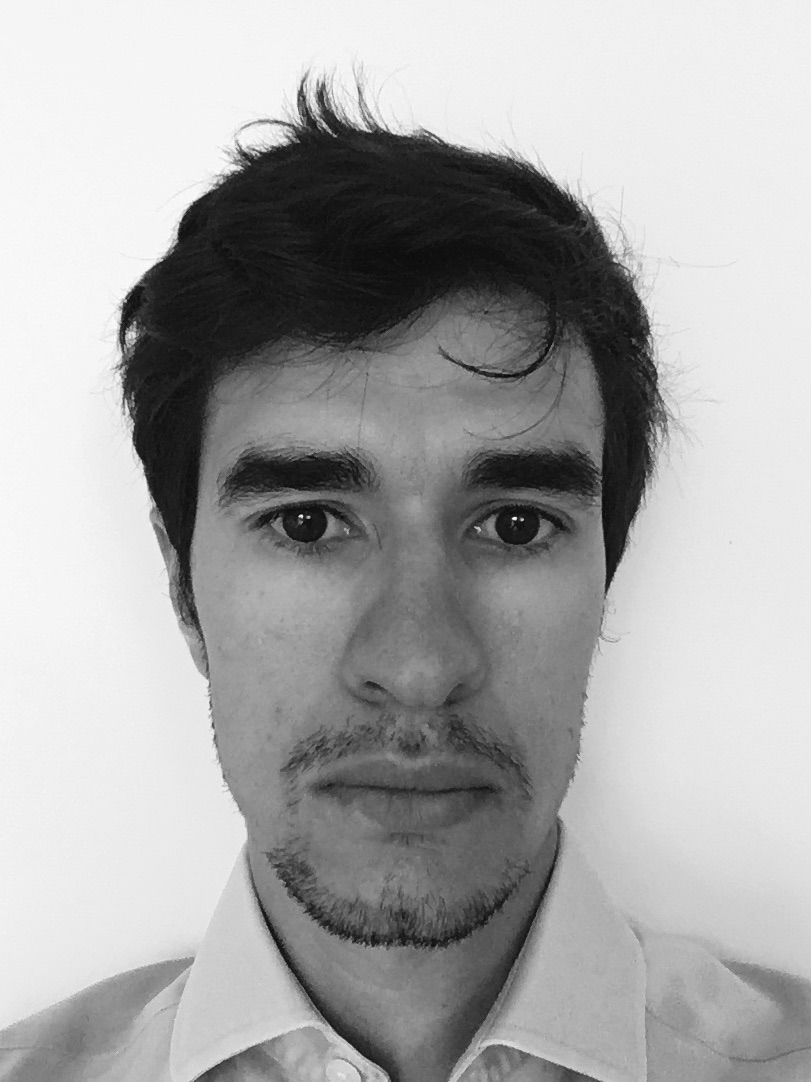}}]{Louis Grimaud} is a control expert in Safran Electronics and Defence (SED) in the field of power electronic and electrical systems. He graduated from ESCPE Lyon in 2011 and joined SED to work on LEAP engines regulation system. 
He joined the Research $\&$ Technology Department in 2015 to led several projects on modelling and robust control of soft-switching high frequency DC-DC converters, PFC and multi-level inverter. 
His current research covers control and observation of nonlinear system for critical real-time applications with an emphasis on mastering harmonics for aircraft electrical propulsion and actuation systems. He holds several patents on robust control, power converter topologies and electrical architecture.
 \end{IEEEbiography}
 \begin{IEEEbiography}[{\includegraphics[width=1in,height=1.25in,clip,keepaspectratio]{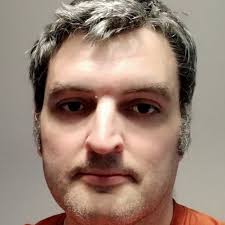}}]{Philippe Feyel} is a R$\&$T engineer for the high-tech company Safran Electronics $\&$ Defense (Safran Group). Senior Expert in automation applied to the line of sight stabilization problem and with a PhD in Automation sciences, he works in partnership with the academic world on the industrial implementation of robust and intelligent control through using modern optimization techniques and artificial intelligence tools (stochastic optimization by meta heuristics, non-smooth optimization, dynamic neural networks and machine learning).\end{IEEEbiography}
\end{document}